\documentclass{sig-alternate-10}

\usepackage{amsmath}
\usepackage{amsfonts}
\usepackage{epsfig}
\usepackage{graphics}
\usepackage{amssymb}
\usepackage{color}
\usepackage{comment}
\usepackage{paralist}
\usepackage{multirow}
\usepackage{subfigure}
\usepackage{eepic}
\usepackage{stmaryrd}
\usepackage{textcomp}
\usepackage{times}
\usepackage{balance}
\usepackage{ulem}

\begin{document}



\newcommand{\A}{\mathcal{A}} \newcommand{\B}{\mathcal{B}}
\newcommand{\C}{\mathcal{C}} \newcommand{\D}{\mathcal{D}}
\newcommand{\E}{\mathcal{E}} \newcommand{\F}{\mathcal{F}}
\newcommand{\G}{\mathcal{G}} \renewcommand{\H}{\mathcal{H}}
\newcommand{\I}{\mathcal{I}} \newcommand{\J}{\mathcal{J}}
\newcommand{\K}{\mathcal{K}} \renewcommand{\L}{\mathcal{L}}
\newcommand{\M}{\mathcal{M}} \newcommand{\N}{\mathcal{N}}
\renewcommand{\O}{\mathcal{O}} \renewcommand{\P}{\mathcal{P}}
\newcommand{\Q}{\mathcal{Q}} \newcommand{\R}{\mathcal{R}}
\renewcommand{\S}{\mathcal{S}} \newcommand{\T}{\mathcal{T}}
\newcommand{\U}{\mathcal{U}} \newcommand{\V}{\mathcal{V}}
\newcommand{\W}{\mathcal{W}} \newcommand{\X}{\mathcal{X}}
\newcommand{\Y}{\mathcal{Y}} \newcommand{\Z}{\mathcal{Z}}


\newcommand{\setone}[2][1]{\set{#1\cld #2}}
\newcommand{\eset}{\emptyset}
\newcommand{\ol}[1]{\overline{#1}}                
\newcommand{\ul}[1]{\underline{#1}}               
\newcommand{\uls}[1]{\underline{\raisebox{0pt}[0pt][0.45ex]{}#1}}

\newcommand{\ra}{\rightarrow}
\newcommand{\Ra}{\Rightarrow}
\newcommand{\la}{\leftarrow}
\newcommand{\La}{\Leftarrow}
\newcommand{\lra}{\leftrightarrow}
\newcommand{\Lra}{\Leftrightarrow}
\newcommand{\lora}{\longrightarrow}
\newcommand{\Lora}{\Longrightarrow}
\newcommand{\lola}{\longleftarrow}
\newcommand{\Lola}{\Longleftarrow}
\newcommand{\lolra}{\longleftrightarrow}
\newcommand{\Lolra}{\Longleftrightarrow}
\newcommand{\ua}{\uparrow}
\newcommand{\Ua}{\Uparrow}
\newcommand{\da}{\downarrow}
\newcommand{\Da}{\Downarrow}
\newcommand{\uda}{\updownarrow}
\newcommand{\Uda}{\Updownarrow}


\newcommand{\incl}{\subseteq}
\newcommand{\imp}{\rightarrow}
\newcommand{\deq}{\doteq}
\newcommand{\dleq}{\dot{\leq}}                   


\newcommand{\per}{\mbox{\bf .}}                  

\newcommand{\cld}{,\ldots,}                      
\newcommand{\ld}[1]{#1 \ldots #1}                 
\newcommand{\cd}[1]{#1 \cdots #1}                 
\newcommand{\lds}[1]{\, #1 \; \ldots \; #1 \,}    
\newcommand{\cds}[1]{\, #1 \; \cdots \; #1 \,}    

\newcommand{\dd}[2]{#1_1,\ldots,#1_{#2}}             
\newcommand{\ddd}[3]{#1_{#2_1},\ldots,#1_{#2_{#3}}}  
\newcommand{\dddd}[3]{#1_{11}\cld #1_{1#3_{1}}\cld #1_{#21}\cld #1_{#2#3_{#2}}}

\newcommand{\ldop}[3]{#1_1 \ld{#3} #1_{#2}}   
\newcommand{\cdop}[3]{#1_1 \cd{#3} #1_{#2}}   
\newcommand{\ldsop}[3]{#1_1 \lds{#3} #1_{#2}} 
\newcommand{\cdsop}[3]{#1_1 \cds{#3} #1_{#2}} 


\newcommand{\quotes}[1]{{\lq\lq #1\rq\rq}}
\newcommand{\set}[1]{\{#1\}}                      
\newcommand{\Set}[1]{\left\{#1\right\}}
\newcommand{\bigset}[1]{\Bigl\{#1\Bigr\}}
\newcommand{\bigmid}{\Big|}
\newcommand{\card}[1]{|{#1}|}                     
\newcommand{\Card}[1]{\left| #1\right|}
\newcommand{\cards}[1]{\sharp #1}
\newcommand{\sub}[1]{[#1]}
\newcommand{\tup}[1]{\langle #1\rangle}            
\newcommand{\Tup}[1]{\left\langle #1\right\rangle}


\def\qed{\hfill{\qedboxempty}      
  \ifdim\lastskip<\medskipamount \removelastskip\penalty55\medskip\fi}

\def\qedboxempty{\vbox{\hrule\hbox{\vrule\kern3pt
                 \vbox{\kern3pt\kern3pt}\kern3pt\vrule}\hrule}}

\def\qedfull{\hfill{\qedboxfull}   
  \ifdim\lastskip<\medskipamount \removelastskip\penalty55\medskip\fi}

\def\qedboxfull{\vrule height 4pt width 4pt depth 0pt}

\newcommand{\markfull}{\qedboxfull}
\newcommand{\markempty}{\qed}





\newcommand{\assign}{:=}

\newcommand{\varpos}[2]{[#1]_{#2}}


\newcommand{\SA}{{\sf SemAc}}
\newcommand{\FSA}{{\sf FinSemAc}}
\newcommand{\ESA}{{\sf SemAc}$^{=}$}
\newcommand{\FESA}{{\sf FinSemAc}$^{=}$}
\newcommand{\saeval}{{\sf SemAcEval}}
\newcommand{\Cont}{{\sf Cont}}
\newcommand{\RCont}{{\sf RestCont}}
\newcommand{\ABCont}{{\sf AcBoolCont}}
\newcommand{\ac}{\mathbb{AE}}
\newcommand{\full}{\mathbb{F}}
\newcommand{\guarded}{\mathbb{G}}
\newcommand{\linear}{\mathbb{L}}
\newcommand{\id}{\mathbb{ID}}
\newcommand{\nr}{\mathbb{NR}}
\newcommand{\sticky}{\mathbb{S}}
\newcommand{\class}[1]{\mathbb{#1}}

\newcommand{\blank}{\sqcup}

\newcommand{\child}{\mathit{child}}
\newcommand{\pchild}{\mathit{parentorchild}}
\newcommand{\descend}{\mathit{descendant}}
\newcommand{\nroot}{\mathit{root}}
\newcommand{\nct}{\mathit{CT}}
\newcommand{\suc}{\mathit{succ}}
\newcommand{\same}{\mathit{same}}
\newcommand{\samecell}{\mathit{sameCell}}
\newcommand{\sameleaf}{\mathit{sameLeaf}}
\newcommand{\nskel}{s}
\newcommand{\ncell}{c}

\newcommand{\mi}[1]{\mathit{#1}}

\newcommand{\con}{\mathit{conf}}
\newcommand{\beg}{\mathit{begin}}

\newcommand{\enc}[1]{\mathit{enc}(#1)}

\newcommand{\troot}[1]{\mathit{root}(#1)}
\newcommand{\tnodes}[1]{\mathit{nodes}(#1)}
\newcommand{\critical}[1]{\mathit{critical}(#1)}

\newcommand{\fr}[1]{\mathit{frontier}(#1)}
\newcommand{\nonfr}[1]{\mathit{nonfrontier}(#1)}

\newcommand{\tw}[1]{\mathit{tw}(#1)}

\newcommand{\norm}[1]{\mathsf{N}(#1)}

\newcommand{\ovtabentrya}[5]{
  \multirow{3}{*}{#1}&
  \footnotesize {#2} &
  \footnotesize {#3} &
  \footnotesize {#4} &
  \footnotesize {#5}
}

\newcommand{\ovtabentryb}[4]{
  &
  {#1}&
  {#2}&
  {#3}&
  {#4}
}

\newcommand{\ovtabentryc}[4]{
  &
  \footnotesize {#1} &
  \footnotesize {#2} &
  \footnotesize {#3} &
  \footnotesize {#4}
}

\newcommand{\cmpitem}{\noindent\ensuremath{-}\xspace}

\def\naf{\ensuremath{\raise.17ex\hbox{\ensuremath{\scriptstyle\mathtt{\sim}}}}\xspace}
\def\nafd{\ensuremath{\raise.17ex\hbox{\ensuremath{\scriptstyle\mathtt{\sim}.}}}\xspace}

\newcommand{\bas}[3]{\mathsf{B}(#1,#2,#3)}
\newcommand{\base}[2]{\mathsf{B}(#1,#2)}
\newcommand{\skolem}[2]{\mathsf{S}(#1,#2)}
\newcommand{\unify}[2]{\mathsf{U}(#1,#2)}
\newcommand{\match}[2]{\mathit{heads}(#1,#2)}
\newcommand{\cunify}[3]{\mathsf{C}_{#1}(#2,#3)}
\newcommand{\cunifyy}[2]{\mathsf{C}_{#1}(#2)}
\newcommand{\lin}[2]{\mathsf{b}(#1,#2)}
\newcommand{\mgu}[2]{\mathsf{MGU}(#1,#2)}
\newcommand{\f}[1]{\mathsf{F}(#1)}
\newcommand{\ff}[1]{\textsf{ff}(#1)}

\newcommand{\nav}[1]{\mathit{nav}(#1)}

\newcommand{\DB}{\mathit{DB}} \newcommand{\wrt}[0]{w.r.t.\ }

\newcommand{\ifdirection}{``$\Leftarrow$'' {}}
\newcommand{\onlyifdirection}{``$\Rightarrow$'' {}}

\renewcommand{\emptyset}{\varnothing}

\newcommand{\qans}[3]{\mathsf{QAns}(#1,#2,#3)} 

\newcommand{\relevent}[2]{\substack{#1,#2 \\ \twoheadrightarrow}} 


\newcommand{\rel}[1]{\mathsf{#1}}
\newcommand{\attr}[1]{\mathit{#1}}
\newcommand{\const}[1]{\mathit{#1}}
\newcommand{\vett}[1]{\vec{#1}}

\newcommand{\ext}[2]{#1^{#2}}


\newcommand{\pred}[1]{\mathit{pred}(#1)}

\newcommand{\atom}[1]{\underline{#1}}
\newcommand{\tuple}[1]{\mathbf{#1}}

\newcommand{\parent}[1]{\mathit{parent}(#1)}
\newcommand{\dept}[1]{\mathit{depth}(#1)}

\newcommand{\mar}[1]{\hat{#1}}


\newcommand{\dom}{\mathbf{C}}
\newcommand{\freshdom}{\mathbf{N}}
\newcommand{\variables}{\mathbf{V}}
\newcommand{\adom}[1]{\mathit{dom}(#1)}
\newcommand{\var}[1]{\mathit{var}(#1)}
\newcommand{\vari}[2]{\mathit{var}_{#1}(#2)}

\newcommand{\aff}[1]{\mathit{affected}(#1)}
\newcommand{\nonaff}[1]{\mathit{nonaffected}(#1)}
\newcommand{\sch}[1]{\mathit{sch}(#1)}


\newcommand{\entities}{\mathit{Ent}}
\newcommand{\relationships}{\mathit{Rel}}
\newcommand{\attributes}{\mathit{Att}}
\newcommand{\symbols}{\mathit{Sym}}


\newcommand{\dep}{\Sigma}
\newcommand{\tdep}{\Sigma_T}
\newcommand{\edep}{\Sigma_E}
\newcommand{\fdep}{\Sigma_F}
\newcommand{\kdep}{\Sigma_K}
\newcommand{\ndep}{\Sigma_\bot}

\newcommand{\key}[1]{\mathit{key}(#1)}
\newcommand{\isa}[1]{\mathit{ISA}}
\newcommand{\no}{\textrm{not}}


\newcommand{\Var}[1]{\mathit{Var}(#1)}
\newcommand{\head}[1]{\mathit{H}(#1)}
\newcommand{\heads}[1]{\mathsf{heads}(#1)}
\newcommand{\body}[1]{\mathit{B}(#1)}
\newcommand{\pbody}[1]{\mathit{B}^{+}(#1)}
\newcommand{\nbody}[1]{\mathit{B}^{-}(#1)}
\newcommand{\arity}[1]{\mathit{arity}(#1)}
\newcommand{\conj}[1]{\mathit{conj}(#1)}

\newcommand{\dvar}[1]{\mathit{dvar}(#1)}
\newcommand{\rset}[2]{\mathit{rset}(#1,#2)}

\newcommand{\cover}[1]{\mathit{cover}(#1)}

\newcommand{\ans}[3]{\mathit{ans}(#1,#2,#3)}
\newcommand{\pos}[2]{\mathit{pos}(#1,#2)}


\newcommand{\ins}[1]{\mathbf{#1}}
\newcommand{\insA}{\ins{A}}
\newcommand{\insB}{\ins{B}}
\newcommand{\insC}{\ins{C}}
\newcommand{\insD}{\ins{D}}
\newcommand{\insX}{\ins{X}}
\newcommand{\insY}{\ins{Y}}
\newcommand{\insZ}{\ins{Z}}
\newcommand{\insK}{\ins{K}}
\newcommand{\insU}{\ins{U}}
\newcommand{\insT}{\ins{T}}
\newcommand{\insO}{\ins{O}}
\newcommand{\insW}{\ins{W}}
\newcommand{\insN}{\ins{N}}
\newcommand{\insV}{\ins{V}}


\newcommand{\chase}[2]{\mathit{chase}(#1,#2)}
\newcommand{\ochase}[2]{\mathit{Ochase}(#1,#2)}
\newcommand{\rchase}[2]{\mathit{Rchase}(#1,#2)}
\newcommand{\instant}{\mathit{inst}}

\newcommand{\pchase}[3]{\mathit{chase}^{#1}(#2,#3)}
\newcommand{\apchase}[3]{\mathit{chase}^{[#1]}(#2,#3)}
\newcommand{\aprchase}[3]{\mathit{Rchase}^{[#1]}(#2,#3)}
\newcommand{\level}[1]{\mathit{level}(#1)}
\newcommand{\freeze}[1]{\mathit{fr}(#1)}
\newcommand{\mods}[2]{\mathit{mods}(#1,#2)}
\newcommand{\subs}[2]{\gamma_{#1,#2}}


\renewcommand{\paragraph}[1]{\textbf{#1}}
\newenvironment{proofsk}{\textsc{Proof (sketch).}}{\hfill$\Box$\newline}
\newenvironment{pf}{\textsc{Proof.}}{\hfill$\Box$\newline}
\newenvironment{proofrsk}{\textsc{Proof (rough sketch).}}{$\Box$\newline}
\newenvironment{proofidea}{\textsc{Proof idea.}}{$\Box$\newline}


\newcommand{\dg}[2]{\mathit{DG}(#1,#2)}
\newcommand{\rank}[1]{\mathit{rank}(#1)}

\newcommand{\EXP}{{\scshape ExpTime}}
\newcommand{\PTIME}{\textsc{P}}
\newcommand{\NP}{{\scshape NP}}
\newcommand{\co}{{co}}
\newcommand{\LOGSPACE}{\textsc{LogSpace}}
\newcommand{\AC}[1]{\textsc{$\mbox{AC}_{#1}$}}
\newcommand{\NC}[1]{\textsc{$\mbox{NC}_{#1}$}}
\newcommand{\LOGCFL}{\textsc{LogCFL}}
\newcommand{\SIGMA}[2]{$\Sigma_{\textrm{#2}}^{\textrm{#1}}$}
\newcommand{\PI}[2]{$\Pi_{\textrm{#2}}^{\textrm{#1}}$}
\newcommand{\PSPACE}{\textsc{PSpace}}
\newcommand{\LINSPACE}{\textsc{LinSpace}}
\newcommand{\EXPSPACE}{\textsc{ExpSpace}}
\newcommand{\NEXP}{\textsc{NExpTime}}

\newcommand{\datalogpm}{Datalog$^\pm$}

\newcommand{\dllite}[2]{\textit{DL-Lite$^{\cal #1}_{#2}$}}

\newcommand{\dinclusion}{DIDs}
\newcommand{\certain}{certain satisfaction}
\newcommand{\twoexptime}{2\EXP}
\newcommand{\exptime}{\EXP}
\newcommand{\exspace}{\EXPSPACE}
\newcommand{\Or}{\mathrm{Or}}
\newcommand{\False}{\mathrm{False}}
\newcommand{\True}{\mathrm{True}}
\newcommand{\mis}{monadic disjunctive inclusion dependencies}
\newcommand{\dfd}{disjunctive full dependencies}

\newcommand{\Skolem}[1]{{\mathcal{S}(#1)}}

\newcommand{\TODO}[1]{\textbf{[TODO:} #1\textbf{]}}
\newcommand{\malvi}[2]{{\color{red}\sout{#1} #2}}


\newcommand{\OMIT}[1]{}

\newtheorem{theorem}{Theorem}
\newtheorem{corollary}[theorem]{Corollary}
\newtheorem{proposition}[theorem]{Proposition}
\newtheorem{lemma}[theorem]{Lemma}
\newtheorem{claim}[theorem]{Claim}
\newtheorem{fact}[theorem]{Fact}
\newtheorem{apptheorem}{Theorem}[section]
\newtheorem{appcorollary}[apptheorem]{Corollary}
\newtheorem{appproposition}[apptheorem]{Proposition}
\newtheorem{applemma}[apptheorem]{Lemma}
\newtheorem{appclaim}[apptheorem]{Claim}
\newtheorem{appfact}[apptheorem]{Fact}

\newdef{definition}{Definition}
\newdef{example}{Example}
\newdef{appdefinition}{Definition}
\newdef{appexample}{Example}

\title{Semantic Acyclicity Under Constraints}

\numberofauthors{3}

\author{
\alignauthor
Pablo Barcel\'{o}\\
       \affaddr{{\small Center for Semantic Web Research \&}}\\
       \affaddr{{\small DCC, University of Chile}}\\
       \email{{\small pbarcelo@dcc.uchile.cl}}
\alignauthor Georg Gottlob\\
       \affaddr{{\small Dept. of Computer Science}}\\
       \affaddr{{\small University of Oxford}}\\
       \email{{\small georg.gottlob@cs.ox.ac.uk}}
\alignauthor Andreas Pieris\\
       \affaddr{{\small Inst. of Information Systems}}\\
       \affaddr{{\small TU Wien}}\\
       \email{{\small pieris@dbai.tuwien.ac.at}}
}

\maketitle

\sloppy

\fontsize{10pt}{10.2pt}
\selectfont

\begin{abstract}
  A conjunctive query (CQ) is semantically acyclic if it is equivalent
  to an acyclic one. Semantic acyclicity has been studied in the
  constraint-free case, and deciding whether a query enjoys
  this property is \NP-complete. However, in case the
  database is subject to constraints such as tuple-generating
  dependencies (tgds) that can express, e.g., inclusion
  dependencies, or equality-generating dependencies (egds) that
  capture, e.g., functional dependencies, a CQ may turn out
  to be semantically acyclic under the constraints while not
  semantically acyclic in general. This opens avenues to new
  query optimization techniques. In this paper we initiate and
  develop the theory of semantic acyclicity under constraints. More
  precisely, we study the following natural problem: Given a CQ
  and a set of constraints, is the query
  semantically acyclic under the constraints, or, in other
  words, is the query equivalent to an acyclic one over all those
  databases that satisfy the set of constraints?

  We show that, contrary to what one might expect, decidability of CQ containment is a necessary but not sufficient condition for the decidability of semantic acyclicity. In particular, we show that semantic acyclicity is undecidable in the presence of full tgds (i.e., Datalog rules). In view of this fact, we focus on the main classes of tgds for which CQ containment is decidable, and do not capture the class of full tgds, namely guarded, non-recursive and sticky tgds. For these classes we show that semantic acyclicity is decidable, and its complexity coincides with the complexity of CQ containment. In the case of egds, we show that if we focus on keys over unary and binary predicates, then semantic acyclicity is decidable (\NP-complete). We finally consider the problem of evaluating a semantically acyclic query over a database that satisfies a set of constraints. For guarded tgds and functional dependencies the evaluation problem is tractable.
\end{abstract}

\section{Introduction}\label{sec:introduction}

Query optimization is a fundamental database
task that amounts to transforming a query into one that is arguably
more efficient to evaluate. The database theory community has
developed several principled methods
for optimization of conjunctive queries (CQs), many of which are based on {\em
  static-analysis} tasks such as containment \cite{AbHV95}.
In a nutshell, such methods compute a {\em minimal} equivalent version
of a CQ, where minimality refers to number of
atoms. As argued by Abiteboul, Hull, and Vianu \cite{AbHV95}, this
provides a theoretical notion of ``true optimality'' for
the reformulation of a CQ, as opposed to practical considerations based on
heuristics.
For each CQ $q$ the minimal equivalent CQ
is its {\em core} $q'$
\cite{HN}. Although the static analysis tasks that support CQ
minimization are NP-complete
\cite{ChMe77}, this is not a major problem for
real-life applications, as the
input (the CQ) is small.

It is known, on the other hand, that semantic information about the
data, in the form of integrity constraints, alleviates query
optimization by reducing the space of possible reformulations.
In the previous analysis, however,
constraints play no role, as
CQ equivalence is defined over {\em all} databases. Adding
constraints yields a refined notion of CQ equivalence, which
holds over those databases that satisfy a given set of constraints
only. But finding
a minimal equivalent CQ in this context is notoriously
more difficult than before. This is because
basic static analysis tasks such as containment become undecidable
when considered in full generality. This motivated a long research
program for finding larger ``islands of decidability'' of such
containment problem, based on syntactical restrictions on
constraints
\cite{BMRT11,CaGK13,CaGP12,CGL08,JoKl84,KrRu11}.

An important shortcoming of the previous approach, however, is that there is no theoretical guarantee that the minimized version of a CQ is in fact
easier to evaluate (recall that, in general, CQ evaluation is
\NP-complete \cite{ChMe77}).
We know, on the other hand, quite a bit about
classes of CQs that can be evaluated efficiently. It is thus a natural
problem to ask whether constraints can be used to reformulate a
CQ as one in such
tractable classes, and if so, what is the cost of
computing such reformulation. Following Abiteboul et al., this would
provide us with a theoretical guarantee of ``true efficiency'' for those
reformulations. We focus on one of the oldest
and most studied tractability conditions for CQs; namely, {\em
  acyclicity}. It is known that acyclic CQs can be evaluated in linear time
  \cite{Yan81}.

More formally, let us write $q \equiv_\Sigma q'$ whenever CQs $q$ and $q'$ are
equivalent over all databases that satisfy $\Sigma$.
In this work we study the following problem:

\begin{center}
\fbox{\begin{tabular}{ll}
{\small PROBLEM} : & {\sc Semantic Acyclicity} \\
\hline
\\{\small INPUT} : & A CQ $q$ and a finite set $\Sigma$ of constraints.
\\
{\small QUESTION} : &  Is there an acyclic CQ $q'$ s.t.
$q \equiv_\Sigma q'$?
\end{tabular}}
\end{center}

We study this problem for the two
most important classes of database constraints; namely:

\begin{enumerate}

\item {\em Tuple-generating dependencies} (tgds), i.e., expressions of
  the form $\forall \bar x \forall \bar y (\phi(\bar x,\bar y)
  \rightarrow \exists \bar z \psi(\bar x,\bar z))$, where $\phi$ and
  $\psi$ are conjuntions of atoms. Tgds subsume the important class of
  referential integrity constraints (or inclusion dependencies).

\item {\em Equality-generating dependencies} (egds), i.e., expressions of
  the form $\forall \bar x (\phi(\bar x) \rightarrow y = z)$, where
  $\phi$ is a conjunction of atoms and $y,z$ are variables in $\bar
  x$.
Egds subsume
keys and functional dependencies (FDs).

\end{enumerate}
A useful aspect of tgds and egds is that containment under them
can be
studied in terms of the
 {\em chase procedure} \cite{MaMS79}.

Coming back to semantic acyclicity, the main problem we study is, of course, decidability. Since basic reasoning with
tgds and egds is, in general, undecidable, we cannot expect semantic acyclicity to be
decidable for arbitrary such constraints.
Thus, we concentrate on the following question:

\medskip

\noindent
{\em \underline{Decidability:}}
For which classes of tgds and egds is the problem of semantic acyclicity decidable? In
  such cases, what is the computational cost of the problem?

\medskip

Since semantic acyclicity is defined in terms of CQ equivalence under
constraints, and
the latter has received a lot of attention, it is relevant also to study
the following question:

\medskip

\noindent
{\em \underline{Relationship to CQ equivalence:}} What is the relationship between CQ equivalence and semantic
acyclicity under constraints?
Is the latter decidable for each class of tgds and egds for which the former
is decidable?

\medskip

Notice that if this was the case, one could transfer the mature theory of CQ equivalence under tgds and egds
to tackle
the problem of semantic acyclicity.

Finally, we want to understand to what extent semantic
acyclicity helps
CQ evaluation.
Although an acyclic reformulation of a CQ can be evaluated efficiently, computing such reformulation
might be expensive. Thus, it  is relevant to study the following question:

\medskip

\noindent
{\em \underline{Evaluation:}} What is the computational cost
of evaluating semantically acyclic CQs under constraints?

\medskip

\smallskip

\noindent
\paragraph{Semantic acyclicity in the absence of constraints.}
The semantic acyclicity problem in the absence of
dependencies (i.e., checking whether a CQ $q$ is equivalent to an
acyclic one over the set of all databases) is by now well-understood.
Regarding decidability, it is easy to prove that
a CQ $q$ is semantically acyclic iff its core $q'$ is
  acyclic. (Recall that such $q'$ is the minimal equivalent CQ to
  $q$). It follows that checking semantic acyclicity in the absence
  of constraints is \NP-complete (see, e.g., \cite{BRV13}).
Regarding evaluation, semantically acyclic CQs can be
evaluated efficiently \cite{ChenD05,DaKV02,GGM09}.

\medskip

\noindent
\paragraph{The relevance of constraints.}
In the absence of constraints a CQ $q$ is equivalent to an acyclic one
iff its core $q'$ is acyclic. Thus, the only reason why $q$ is not
acyclic in the first hand is because it has not been minimized. This
tells us that in this context semantic acyclicity is not really
different from usual minimization.  The presence of constraints, on
the other hand, yields a more interesting notion of semantic
acyclicity.
This is because constraints can be applied
on CQs to produce acyclic reformulations of them.

\noindent
\begin{example} \label{ex:intro}
This simple example helps understanding the role of tgds when
reformulating CQs as acyclic ones. Consider a database that stores
information about customers, records, and musical styles. The relation
${\tt Interest}$ contains pairs $(c,s)$ such that customer
$c$ has declared
interest in style $s$. The relation ${\tt
  Class}$ contains pairs $(r,s)$ such that record $r$ is
of style $s$. Finally, the relation ${\tt Owns}$
contains a pair $(c,r)$ when customer $c$ owns record $r$.

Consider now a CQ $q(x,y)$ defined as follows:
$$\exists z \big({\tt Interest}(x,z) \wedge {\tt Class}(y,z) \wedge {\tt
  Owns}(x,y)\big).$$
This query asks for pairs $(c,r)$ such that customer $c$ owns record $r$
and has expressed interest in at least one of the styles with which
$r$ is associated. This CQ is a core but it is not acyclic. Thus, from
our previous observations it is not equivalent to an acyclic CQ (in
the absence of constraints).

Assume now that we are told that this database contains compulsive
music collectors only. In particular, each customer owns every record that is
classified with a style in which he/she has expressed interest. This
means that the database satisfies the tgd:
$$\tau \ = \ {\tt Interest}(x,z),{\tt Class}(y,z) \, \rightarrow \, {\tt
  Owns}(x,y).$$
With this information at hand, we can easily reformulate $q(x,y)$
as the following acyclic CQ $q'(x,y)$:
$$\exists z \big({\tt Interest}(x,z) \wedge {\tt Class}(y,z)\big).$$
Notice that $q$ and $q'$ are in fact equivalent over every database
that satisfies $\tau$. \hfill\markfull
\end{example}

\noindent
\paragraph{Contributions.} We observe that semantic acyclicity under
constraints is not only more powerful, but also theoretically more
challenging than in the absence of them. We start by studying
decidability. In the process we also clarify the relationship between
CQ equivalence and semantic acyclicity.

\medskip

\noindent
{\em \underline{Results for tgds:}}
Having a decidable CQ containment problem
is a necessary condition for semantic acyclicity to be decidable under tgds.\footnote{
Modulo some mild technical
assumptions elaborated in the paper.}
Surprisingly enough, it is not a sufficient condition. This means
that, contrary to what one might expect, there are natural classes of
tgds for which CQ containment but not semantic acyclicity is
decidable. In particular, this is the case for the well-known
class of {\em full} tgds (i.e., tgds without
  existentially quantified variables in the head).
In conclusion, we cannot directly export techniques from
  CQ containment to deal with semantic acyclicity.

In view of the previous results, we concentrate on classes of tgds
  that (a) have a decidable CQ containment problem, and (b) do not
  contain the class of full tgds. These restrictions are satisfied by
  several expressive languages considered in the literature. Such languages can be
  classified into three main families depending on the techniques used
  for studying their containment problem: (i) {\em guarded} tgds
  \cite{CaGK13}, which contain inclusion and linear dependencies,
(ii) {\em non-recursive} \cite{FKMP05}, and (iii) {\em sticky}
sets of tgds \cite{CaGP12}.
Instead of studying such languages one by one, we identify two semantic
criteria that yield decidability for the semantic acyclicity problem, and then
show that each one of the languages satisfies one such criteria.

\begin{itemize}
\item The first criterion is {\em acyclicity-preserving chase}.
This is satisfied by those tgds for which the application of the chase
over an acyclic instance preserves acyclicity. Guarded tgds enjoy
this property. We establish that semantic acyclicity
under guarded tgds is decidable and has the same complexity
than its associated CQ containment problem: \twoexptime-complete, and
\NP-complete for a fixed schema.


\item The second criterion is {\em rewritability by unions of CQs (UCQs)}. Intuitively,
a class $\class{C}$ of sets of tgds has this property if the
CQ containment problem under a set in $\class{C}$ can
always be reduced to a UCQ
containment problem without constraints.
Non-recursive and sticky sets of tgds enjoy this property. In the
first case the complexity
matches that of its associated CQ containment problem: \NEXP-complete, and
\NP-complete if the schema is fixed. In the second case, we get a \NEXP\ upper bound and an \exptime\ lower bound. For a fixed schema the problem is \NP-complete.
\end{itemize}

The \NP\ bounds (under a fixed schema) can be seen as positive results:
By spending exponential time in the size of the (small)
query, we can not only minimize it using known techniques but also
find an acyclic reformulation if one exists.

\medskip

\noindent
{\em \underline{Results for egds:}} After showing that the techniques developed for tgds cannot be applied for showing the decidability of semantic acyclicity under egds, we focus on the class of keys over unary and binary predicates and we establish a positive result, namely semantic acyclicity is \NP-complete. We prove this by showing that in such context keys have acyclicity-preserving chase. Interestingly, this positive result can be extended to unary functional dependencies (over unconstrained signatures); this result has been established independently by Figueira~\cite{Figu16}. We leave open whether the problem of semantic acyclicity under arbitrary egds, or even keys over arbitrary schemas, is decidable.

\medskip

\noindent {\em \underline{Evaluation:}} For
tgds for which semantic acyclicity is decidable (guarded,
non-recursive, sticky), we can use the following algorithm to evaluate
a semantically acyclic CQ $q$ over a database $D$ that satisfies the
constraints $\Sigma$: 
\begin{enumerate}\itemsep-\parsep
\item Convert $q$ into an equivalent acyclic CQ $q'$ under $\Sigma$.

\item Evaluate $q'$ on $D$.

\item Return $q(D) = q'(D)$.  
\end{enumerate}
The running time is $O(|D| \cdot f(|q|,|\dep|))$,
where $f$ is a double-exponential function (since $q'$ can be computed
in double-exponential time for each one of the classes
mentioned above and acyclic CQs can be evaluated in linear time).
This constitutes a {\em fixed-parameter tractable algorithm} for evaluating $q$ on $D$. No such algorithm is believed to exist for
CQ evaluation~\cite{PY99}; thus, semantically acyclic CQs under
these constraints behave better than the general case in terms
of evaluation.

But in the absence of constraints one can do better: Evaluating
semantically acyclic CQs in such context is in polynomial time. It is
natural to ask if this also holds in the presence of
constraints. This is the case for guarded tgds and (arbitrary) FDs. For the other classes of constraints the problem remains to be investigated.

\medskip

\noindent
{\em \underline{Further results:}} The results mentioned above continue
to hold for a more ``liberal'' notion based on UCQs, i.e.,
checking whether a UCQ is equivalent
to an acyclic union of CQs under the decidable classes of constraints
identified above.
Moreover, in case that a CQ $q$ is not equivalent to an
acyclic CQ $q'$ under a set of constraints $\Sigma$, our proof techniques yield an {\em approximation of $q$ under $\Sigma$}~\cite{BaLR14},
that is, an acyclic CQ $q'$ that is maximally contained in $q$ under $\Sigma$. Computing and evaluating such approximation yields ``quick'' answers to $q$ when exact evaluation is infeasible.

\medskip

\noindent
\paragraph{Finite vs. infinite databases.} The results mentioned above
interpret the notion of CQ equivalence (and, thus, semantic
acyclicity) over the set of both {\em finite} and {\em infinite}
databases. The reason is the wide application of the chase we make in
our proofs, which characterizes CQ equivalence under arbitrary
databases only. This does not present a serious problem though, as
all the particular classes of tgds for which we prove decidability in the paper
(i.e., guarded, non-recursive, sticky) are {\em finitely
  controllable} \cite{BaGO14,GoMa13}.
This means that CQ equivalence under arbitrary
databases and under finite databases coincide. In conclusion,
the results we obtain
for such classes can be directly exported to the finite case.

\medskip

\noindent
\paragraph{Organization.} Preliminaries are in Section \ref{sec:prelim}.
In Section \ref{sec:tgds} we consider semantic acyclicity
under tgds. Acyclicity-preserving chase is studied in Section \ref{sec:apc},
and UCQ-rewritability in Section \ref{sec:ucq}.
Semantic acyclicity under egds is investigated in Section
\ref{sec:semac-egds}. Evaluation of semantically acyclic CQs is in Section \ref{sec:evaluation}. Finally, we present further advancements in Section \ref{sec:further-advancements} and conclusions
in Section \ref{sec:conclusions}. 
\section{Preliminaries}\label{sec:prelim}

\noindent
\paragraph{Databases and conjunctive queries.}
Let $\insC$, $\insN$ and $\insV$ be disjoint countably infinite sets of {\em
  constants}, {\em (labeled) nulls} and (regular) {\em variables}
(used in queries and dependencies), respectively, and $\sigma$
a relational schema. An {\em atom} over $\sigma$ is an
expression of the form $R(\bar v)$, where $R$ is a relation symbol in
$\sigma$ of arity $n > 0$ and $\bar v$ is an $n$-tuple over $\insC \cup \insN \cup \insV$. An {\em instance} over $\sigma$ is a (possibly infinite) set of atoms over $\sigma$ that contain constants and nulls, while a {\em database} over $\sigma$ is simply a finite instance over $\sigma$.

One of the central notions in our work is acyclicity. An instance $I$ is acyclic if it admits a {\em join tree}; i.e., if there exists a tree $T$ and a mapping $\lambda$ that associates with each node $t$ of $T$
an atom $\lambda(t)$ of $I$, such
that the following holds:
\begin{enumerate}
\item For each atom $R(\bar v)$ in $I$ there is a node $t$ in $T$ such
that $\lambda(t) = R(\bar v)$; and

\item For each null $x$ occurring in $I$ it is the
case that the set $\{t \mid x \in \lambda(t)\}$ is connected in $T$.
\end{enumerate}

A {\em conjunctive query} (CQ)
over $\sigma$ is a formula of the form:
\begin{equation}
\label{eq:cq}
q(\bar x) \ :=
\
\exists \bar y \big(R_1(\bar v_1) \wedge \dots \wedge R_m(\bar
v_m)\big),
\end{equation}
where each $R_i(\bar v_i)$ ($1 \leq i \leq m$) is an atom without nulls
over $\sigma$, each variable mentioned in the $\bar v_i$'s
appears either in $\bar x$ or $\bar y$, and $\bar x$ are the free variables of $q$.
If $\bar x$ is empty, then $q$ is a \emph{Boolean CQ}.
As usual, the evaluation of CQs is defined in terms of {\em
  homomorphisms}. Let $I$ be an instance and $q(\bar x)$ a CQ of the
form $\eqref{eq:cq}$. A homomorphism from $q$ to $I$ is a mapping $h$, which is the identity on $\insC$, from
the variables and constants in $q$ to
the set of constants and nulls $\insC \cup \insN$ such that $R_i(h(\bar v_i)) \in I$,\footnote{As usual, we write
$h(v_1,\dots,v_n)$ for $(h(v_1),\dots,h(v_n))$.} for each $1 \leq i \leq m$. The
{\em evaluation of $q(\bar x)$ over $I$}, denoted $q(I)$, is
the set of all tuples $h(\bar x)$ over $\insC \cup \insN$ such that
$h$ is a homomorphism from $q$ to $I$.

It is well-known that {\em CQ
evaluation}, i.e., the problem of determining if a particular tuple
$\bar t$ belongs to the evaluation $q(D)$
of a CQ $q$ over a database $D$, is \NP-complete \cite{ChMe77}. On the
other hand, CQ evaluation becomes tractable by restricting the
syntactic shape of CQs.
One of the oldest and most common such restrictions is
{\em acyclicity}. Formally, a
CQ $q$ is acyclic if the instance consisting of the atoms of $q$ (after replacing each variable in $q$ with a fresh null) is acyclic. It is known from the seminal work of Yannakakis~\cite{Yan81}, that the problem of evaluating an acyclic CQ $q$ over a database $D$ can be solved in linear time $O(|q| \cdot |D|)$.

\medskip

\noindent
\paragraph{Tgds and the chase procedure.} A {\em tuple-generating dependency} (tgd)
over $\sigma$ is an expression of the form:
\begin{equation}
\label{eq:tgd}
\forall \bar x \forall \bar y \big(\phi(\bar x,\bar y)
\rightarrow \exists \bar z \psi(\bar x,\bar z)\big),
\end{equation}
where both $\phi$ and $\psi$ are conjunctions of atoms without nulls over $\sigma$. For simplicity, we write this tgd as $\phi(\bar x,\bar y)
\rightarrow \exists \bar z \psi(\bar x,\bar z)$, and use comma instead of $\wedge$ for conjoining atoms. Further, we assume
that each variable in $\bar x$ is mentioned in some atom of $\psi$.
We call $\phi$ and $\psi$ the {\em body} and {\em head} of the tgd,
respectively.
The tgd in \eqref{eq:tgd}
is logically equivalent to the expression
$\forall \bar x (q_\phi(\bar x) \rightarrow q_\psi(\bar x))$, where $q_\phi(\bar x)$
and $q_\psi(\bar x)$ are the CQs $\exists \bar y \phi(\bar x,\bar y)$ and
$\exists \bar z \psi(\bar x,\bar z)$, respectively.
Thus, an instance $I$ over $\sigma$ satisfies this tgd
if and only if $q_\phi(I)
\subseteq q_\psi(I)$. We say that an instance $I$ satisfies a set
$\Sigma$ of tgds, denoted $I \models \Sigma$, if $I$ satisfies every
tgd in $\Sigma$.

The {\em chase} is a useful tool when reasoning with tgds \cite{CaGK13,FKMP05,JoKl84,MaMS79}. We start by defining a single chase step.
Let $I$ be an instance over schema $\sigma$ and $\tau = \phi(\bar x,\bar y)
\rightarrow \exists \bar z \psi(\bar x,\bar z)$ a tgd over
$\sigma$.
We say that $\tau$ is \emph{applicable} w.r.t.~$I$ if there exists a tuple $(\bar a,\bar b)$ of elements in $I$ such that $\phi(\bar a,\bar b)$ holds in $I$. In this case, {\em the result of applying $\tau$ over $I$ with $(\bar a,\bar b)$} is the instance $J$ that extends $I$ with every atom in $\psi(\bar a,\bar z')$, where $\bar z'$ is the tuple obtained by simultaneously replacing each variable $z \in \bar z$ with a fresh distinct null not occurring in $I$. For such a single chase step we write $I \xrightarrow{\tau,(\bar a,\bar b)} J$.

Let us assume now that $I$ is an instance and
$\Sigma$ a finite set of tgds. A {\em chase sequence for $I$ under $\Sigma$}
is a sequence:
\[
I_0 \xrightarrow{\tau_0,\bar c_0} I_1 \xrightarrow{\tau_1,\bar c_1} I_2 \dots
\]
of chase steps such that: (1) $I_0 = I$; (2) For each $i \geq 0$, $\tau_i$ is a tgd in $\Sigma$; and (3) $\bigcup_{i \geq 0} I_i \models \Sigma$.
We call $\bigcup_{i \geq 0} I_i$ the {\em result} of this
chase sequence, which always exists.
Although the result of a chase sequence is not necessarily unique (up to isomorphism), each such result is equally useful for our purposes since it can be homomorphically embedded into every other result. Thus, from now on, we denote by $\chase{I}{\Sigma}$ the result of an arbitrary chase
sequence for $I$ under $\Sigma$. Further, for a CQ $q = \exists \bar y \big(R_1(\bar v_1) \wedge \dots \wedge R_m(\bar
v_m)\big)$, we denote by $\chase{q}{\Sigma}$ the result of a chase sequence for the database $\{R_1(\bar v'_1),\dots,R_m(\bar v'_m)\}$ under $\Sigma$ obtained after replacing each variable $x$ in $q$ with a fresh constant $c(x)$.

\medskip

\noindent
\paragraph{Egds and the chase procedure.} An {\em equality-generating dependency} (egd)
over $\sigma$ is an expression of the form:
\begin{equation*}
\label{eq:egd}
\forall \bar x \big(\phi(\bar x) \rightarrow x_i = x_j\big),
\end{equation*}
where $\phi$ is a conjunction of atoms without nulls over $\sigma$, and $x_i,x_j \in \bar x$. For clarity, we write this egd as $\phi(\bar x)
\rightarrow x_i = x_j$, and use comma for conjoining atoms. We call $\phi$ the {\em body} of the egd.
An instance $I$ over $\sigma$ satisfies this egd
if, for every homomorphism $h$ such that $h(\phi(\bar x)) \subseteq I$, it is the case that
$h(x_i) = h(x_j)$. An instance $I$ satisfies a set
$\Sigma$ of egds, denoted $I \models \Sigma$, if $I$ satisfies every
egd in $\Sigma$.

Recall that egds subsume functional dependencies, which in turn subsume keys. A \emph{functional dependency} (FD) over $\sigma$ is an expression of the form $R : A \ra B$, where $R$ is a relation symbol in $\sigma$ of arity $n > 0$, and $A,B$ are subsets of $\{1,\ldots,n\}$, asserting that the values of the attributes of $B$ are determined by the values of the attributes of $A$. For example, $R : \{1\} \ra \{3\}$, where $R$ is a ternary relation, is actually the egd $R(x,y,z),R(x,y',z') \ra z=z'$.
A FD $R : A \ra B$ as above is called \emph{key} if $A \cup B = \{1,\ldots,n\}$.

As for tgds, the chase is a useful tool when reasoning with egds. Let us first define a single  chase step.
Consider an instance $I$ over schema $\sigma$ and an egd $\epsilon = \phi(\bar x) \rightarrow x_i = x_j$ over $\sigma$.
We say that $\epsilon$ is \emph{applicable} w.r.t.~$I$ if there exists a homomorphism $h$ such that $h(\phi(\bar x)) \subseteq I$ and $h(x_i) \neq h(x_j)$. In this case, {\em the result of applying $\epsilon$ over $I$ with $h$} is as follows: If both $h(x_i), h(x_j)$ are constants, then the result is ``failure''; otherwise, it is the instance $J$ obtained from $I$ by identifying $h(x_i)$ and $h(x_j)$ as follows: If one is a constant, then every occurrence of the null is replaced by the constant, and if both are nulls, the one is replaced everywhere by the other.
As for tgds, we can define the notion of the chase sequence for an instance $I$ under a set $\Sigma$ of egds. Notice that such a sequence, assuming that is not failing, always is finite. Moreover, it is unique (up to null renaming), and thus we refer to {\em the} chase for $I$ under $\dep$, denoted $\chase{I}{\dep}$.
Further, for a CQ $q = \exists \bar y \big(R_1(\bar v_1) \wedge \dots \wedge R_m(\bar v_m)\big)$, we denote by $\chase{q}{\Sigma}$ the result of a chase sequence for the database $\{R_1(\bar v'_1),\dots,R_m(\bar v'_m)\}$ under $\Sigma$ obtained after replacing each variable $x$ in $q$ with a fresh constant $c(x)$; however, it is important to clarify that these are special constants, which are treated as nulls during the chase.

\medskip
\noindent
\paragraph{Containment and equivalence.}
Let $q$ and $q'$ be CQs and $\Sigma$ a finite set of tgds or egds. Then, $q$ is {\em contained} in $q'$ under $\Sigma$, denoted $q \subseteq_{\Sigma} q'$, if $q(I) \subseteq q'(I)$ for every instance $I$ such that $I \models \Sigma$.
Further, $q$ is {\em equivalent} to $q'$ under
$\Sigma$, denoted $q \equiv_\Sigma q'$, whenever $q \subseteq_\Sigma q'$ and $q' \subseteq_\Sigma q$ (or, equivalently, if $q(I) = q'(I)$ for every instance $I$ such that $I \models \Sigma$).
The following well-known characterization of CQ containment in terms of the chase will be widely used in our proofs:

\begin{lemma} \label{lemma:cq-equiv-tgds}
Let $q(\bar x)$ and $q'(\bar x')$ be CQs and $\dep$ be a finite set of tgds or egds. Then $q \subseteq_{\Sigma} q'$ if and only if $c(\bar x)$ belongs to the evaluation of $q'$ over $\chase{q}{\Sigma}$.
\end{lemma}

A problem that is quite important for our work is {\em CQ containment under constraints} (tgds or egds), defined as follows: Given CQs $q,q'$ and a finite set $\Sigma$ of tgds or egds, is it the case that $q \subseteq_\Sigma q'$? Whenever $\Sigma$ is bound to belong to a particular class $\class{C}$
of sets of tgds, we denote this problem as $\Cont(\class{C})$.
It is clear that the above lemma provides a decision procedure for the containment problem under egds.
However, this is not the case for tgds.

\medskip
\noindent
\paragraph{Decidable containment of CQs under tgds.}
It is not surprising that Lemma \ref{lemma:cq-equiv-tgds} does not provide a decision procedure for solving CQ containment under tgds since this problem is known to be undecidable~\cite{BeVa81}. This has led to a flurry of activity for identifying syntactic restrictions on sets of tgds that lead to
decidable CQ containment (even in the case when the chase does not terminate).\footnote{In fact, these restrictions are designed to obtain decidable {\em query answering under tgds}. However, this problem is equivalent to query containment under tgds (Lemma~\ref{lemma:cq-equiv-tgds}).} Such restrictions are often
classified into three main paradigms:

\medskip
\noindent
\underline{{\em Guardedness:}} A tgd is {\em guarded} if its
body contains an atom, called {\em guard},
that contains all the body-variables.
Although the chase under guarded tgds does not necessarily terminate, query containment is decidable. This follows from the fact that the result of the chase has {\em bounded treewidth}. Let $\guarded$ be the class of sets of guarded tgds.

\begin{proposition} \label{prop:qa-guarded}{\em \cite{CaGK13}}
$\Cont(\guarded)$ is \twoexptime-complete. It becomes
\EXP-complete if the arity of the schema is fixed, and \NP-complete if the schema is fixed.
\end{proposition}

A key subclass of guarded tgds is the class of {\em linear} tgds, that
is, tgds whose body consists of a single atom~\cite{CaGL12}, which in
turn subsume the well-known class of {\em inclusion dependencies}
(linear tgds without repeated variables neither in the body nor in the
head)~\cite{Fagin81}. Let $\linear$ and $\id$ be the classes of sets
of linear tgds and inclusions dependencies, respectively. $\Cont(\class{C})$, for $\class{C} \in \{\linear,\id\}$, is \PSPACE-complete, and \NP-complete if the arity of the schema is fixed~\cite{JoKl84}.

\medskip \noindent \underline{{\em Non-recursiveness:}} A set
$\Sigma$ of tgds is {\em non-recursive} if its predicate graph contains
no directed cycles. (Non-recursive sets of tgds are also known as {\em
acyclic}~\cite{FKMP05,LMPS15}, but we reserve this term for CQs). Non-recursiveness ensures the termination of the chase, and thus decidability of CQ containment. Let $\nr$ be the class of non-recursive sets of tgds. Then:

\begin{proposition} \label{prop:qa-nr}{\em \cite{LMPS15}}
$\Cont(\nr)$ is complete for \NEXP, even if the arity of the schema is fixed. It becomes \NP-complete if the schema is fixed.
\end{proposition}

\begin{figure}[t]
 \epsfclipon
  \centerline
  {\hbox{
  \leavevmode
  \epsffile{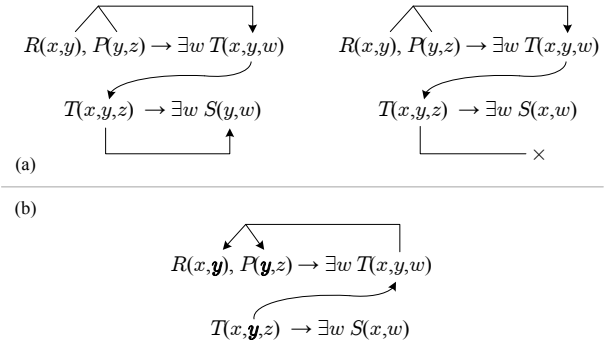}
  }} \epsfclipoff \caption{Stickiness and marking.}
  \label{fig:stickiness}
  \vspace{-2mm}
\end{figure}

\noindent
\underline{{\em Stickiness:}}
This condition ensures neither termination nor bounded treewidth of the chase. Instead, the decidability of query containment is obtained by exploiting {\em query rewriting} techniques.
The goal of stickiness is to capture joins among variables that are not expressible via guarded tgds, but without forcing the chase to terminate.
The key property underlying this condition can be described as follows: During the chase, terms that are associated (via a homomorphism) with variables that appear more than once in the body of a tgd (i.e., join variables) are always propagated (or ``stick'') to the inferred atoms.
This is illustrated in Figure~\ref{fig:stickiness}(a); the first set of tgds is sticky, while the second is not.
The formal definition is based on an inductive marking procedure that marks the variables that may violate the semantic property of the chase described above~\cite{CaGP12}.
Roughly, during the base step of this procedure, a variable that appears in the body of a tgd $\tau$ but not in every head-atom of $\tau$ is marked. Then, the marking is inductively propagated from head to body as shown in Figure~\ref{fig:stickiness}(b).
Finally, a finite set of tgds $\Sigma$ is {\em sticky} if no tgd in $\Sigma$ contains two occurrences of a marked variable. Then:

\begin{proposition} \label{prop:qa-sticky}{\em \cite{CaGP12}}
$\Cont(\sticky)$ is \EXP-complete. It becomes \NP-complete if the arity of the schema is fixed.
\end{proposition}

\noindent \underline{{\em Weak versions:}}
Each one of the previous
classes has an associated weak version, called {\em weakly-guarded}~\cite{CaGK13}, {\em weakly-acyclic}~\cite{FKMP05}, and {\em
weakly-sticky}~\cite{CaGP12}, respectively, that guarantees the decidability of query containment. The underlying idea of all these classes is the same: Relax the conditions in the definition of the class, so that only those positions that receive null values during the chase procedure are taken into
consideration. A key property of all these classes is that they
extend the class of {\em full tgds}, i.e., those without existentially
quantified variables. This is not the case for the ``unrelaxed'' versions presented above. 
\section{Semantic Acyclicity with TGDs}
\label{sec:tgds}

One of the main tasks of our work is to study the problem of checking whether a CQ $q$ is equivalent to an acyclic CQ over those instances that satisfy a set $\Sigma$ of tgds. When this is the case we say that $q$ is {\em semantically acyclic under $\Sigma$}. The semantic acyclicity problem is defined below; $\class{C}$ is a class of sets of tgds (e.g., guarded, non-recursive, sticky, etc.):

\begin{center}
\fbox{\begin{tabular}{ll}
{\small PROBLEM} : & \SA($\class{C}$)
\\{\small INPUT} : & A CQ $q$ and a finite set $\Sigma$ of
tgds in $\class{C}$.
\\
{\small QUESTION} : &  Is there an acyclic CQ $q'$ s.t.
$q \equiv_\Sigma q'$?
\end{tabular}}
\end{center}

\subsection{Infinite Instances vs. Finite Databases}

It is important to clarify that $\SA(\class{C})$ asks for the existence of an acyclic CQ $q'$ that is equivalent to $q$ under $\dep$ focussing on arbitrary (finite or infinite) instances. However, in practice we are concerned only with finite databases. Therefore, one may claim that the natural problem to investigate is $\FSA(\class{C})$, which accepts as input a CQ $q$ and a finite set $\dep \in \class{C}$ of tgds, and asks whether an acyclic CQ $q'$ exists such that $q(D) = q'(D)$ for every finite database $D \models \dep$.

Interestingly, for all the classes of sets of tgds discussed in the
previous section, $\SA$ and $\FSA$ coincide due to the fact that they
ensure the so-called \emph{finite controllability} of CQ containment.
This means that query containment under arbitrary instances and query
containment under finite databases are equivalent problems.
For non-recursive and weakly-acyclic sets of tgds this immediately
follows from the fact that the chase terminates. For guarded-based
classes of sets of tgds this has been shown in~\cite{BaGO14}, while
for sticky-based classes of sets of tgds it has been shown
in~\cite{GoMa13}.
Therefore, assuming that $\class{C}$ is one of the above syntactic
classes of sets of tgds, by giving a solution to $\SA(\class{C})$ we
immediately obtain a solution for $\FSA(\class{C})$.

The reason why we prefer to focus on $\SA(\class{C})$, instead of $\FSA(\class{C})$, is given by Lemma~\ref{lemma:cq-equiv-tgds}: Query containment under arbitrary instances can be characterized in terms of the chase. This is not true for CQ containment under finite databases simply because the chase is, in general, infinite.

\subsection{Semantic Acyclicity vs. Containment}

There is a close relationship between semantic acyclicity and a restricted version of CQ containment under sets of tgds, as we explain next. But first we need to recall the notion of connectedness for queries and tgds. A CQ is {\em connected} if its {\em Gaifman graph} is connected -- recall that the nodes of the Gaifman graph of a CQ $q$ are the variables of $q$, and there is an edge between variables $x$ and $y$ iff they appear together in some atom of $q$.
Analogously, a tgd $\tau$ is \emph{body-connected} if its body is connected. Then:

\begin{proposition} \label{prop:gral-reduction}
Let $\Sigma$ be a finite set of body-connected tgds and $q,q'$ two Boolean and connected CQs without common variables, such that $q$ is acyclic and $q'$ is not semantically acyclic under $\Sigma$. Then $q \subseteq_\Sigma q'$
iff $q \wedge q'$ is semantically acyclic under $\Sigma$.
\end{proposition}

As an immediate corollary of Proposition \ref{prop:gral-reduction}, we obtain an initial boundary for the decidability of $\SA$: We can only obtain a positive result for those classes of sets of tgds for which the restricted containment problem presented above is decidable. More formally, let us define $\RCont(\class{C})$ to be the problem of checking $q \subseteq_\Sigma q'$, given a set $\Sigma$ of body-connected tgds in $\class{C}$ and two Boolean and connected CQs $q$ and $q'$, without common variables, such that $q$ is acyclic and $q'$ is not semantically acyclic under $\Sigma$. Then:

\begin{corollary} \label{coro:easy-undec}
$\SA(\class{C})$ is undecidable for every class $\class{C}$ of tgds such that $\RCont(\class{C})$ is undecidable.
\end{corollary}

As we shall discuss later, $\RCont$ is not easier than general CQ containment under tgds, which means that the only classes of tgds for
which we know the former problem to be decidable are those for which
we know CQ containment to be decidable (e.g., those introduced in
Section \ref{sec:prelim}).

At this point, one might be tempted to think that some version of the
converse of Proposition \ref{prop:gral-reduction} also holds; that is,
the semantic acyclicity problem for $\class{C}$ is reducible to the containment problem for $\class{C}$. This would imply the decidability
of $\SA$ for any class of sets of tgds for which the CQ containment problem is decidable. Our next result shows that the picture is more complicated than this as $\SA$ is undecidable even
over the class $\full$ of sets of full tgds, which ensures the decidability of CQ containment:

\begin{theorem} \label{theo:full-undec}
The problem \SA{\em ($\full$)} is undecidable.
\end{theorem}

\begin{proof}
We provide a sketch
since the complete construction is long.
We reduce from the {\em Post correspondence problem} (PCP)
over the alphabet $\{a,b\}$. The input to this problem are two equally long lists $w_1,\dots,w_n$ and $w'_1,\dots,w'_n$
of words over $\{a,b\}$, and we ask whether there is a {\em solution}, i.e., a nonempty sequence
$i_1 \dots i_m$ of indices in $\{1,\dots,n\}$ such that $w_{i_1} \dots w_{i_m} =
w'_{i_1} \dots w'_{i_m}$.

Let $w_1,\dots,w_n$ and $w'_1,\dots,w'_n$ be an instance of PCP.
In the full proof we construct a Boolean CQ $q$ and a
set $\Sigma$ of full tgds over the signature $\{P_a,P_b,P_\#,P_*,{\rm sync},{\rm start},{\rm end}\}$, where $P_a$, $P_b$, $P_\#$, $P_*$ and ${\rm sync}$ are binary predicates, and ${\rm start}$ and ${\rm end}$ are unary predicates, such that the PCP instance given by $w_1,\dots,w_n$ and $w'_1,\dots,w'_n$ has a solution iff there exists an acyclic CQ $q'$ such that $q \equiv_\Sigma q'$. In this sketch though, we concentrate on the case when the underlying graph of $q'$ is a directed path; i.e, we prove that the PCP instance has a solution iff there is a CQ $q'$ whose underlying graph is a directed path such that $q \equiv_\Sigma q'$. This does not imply the undecidability of the general case, but the proof of the latter is a generalization of the one we sketch below.

The restriction of the query $q$ to the symbols that
are not ${\rm sync}$ is graphically depicted in Figure
\ref{fig:undec1}.
\begin{figure}[t]
 \epsfclipon
  \centerline
  {\hbox{
  \leavevmode
  \epsffile{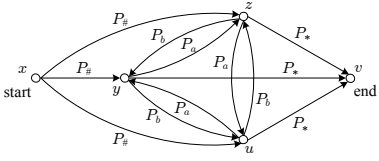}
  }} \epsfclipoff \caption{The query $q$ from the proof of Theorem \ref{theo:full-undec}.}
  \label{fig:undec1}
  \vspace{-2mm}
\end{figure}
There, $x,y,z,u,v$ denote the names of the respective variables. The
interpretation of ${\rm sync}$ in $q$ consists of all pairs in
$\{y,u,z\}$.

Our set $\Sigma$ of full tgds defines the {\em synchronization}
predicate ${\rm sync}$ over those acyclic CQs $q'$ whose underlying graph is a path. Assume that $q'$  encodes a word $w \in \{a,b\}^+$. We denote by $w[i]$, for $1 \leq i \leq |w|$, the prefix of $w$ of length $i$. In such case,
the predicate ${\rm sync}$ contains those pairs $(i,j)$ such that for some sequence $i_1 \dots i_m$ of indices in $\{1,\dots,n\}$ we have that
$w_{i_1} \dots w_{i_m} = w[i]$ and $w'_{i_1} \dots w'_{i_m} = w[j]$. Thus, if $w$ is a solution for the PCP instance, then $(|w|,|w|)$ belongs to the interpretation of ${\rm sync}$.

Formally, $\Sigma$ consists of the following rules:
\begin{enumerate}
\item An initialization rule:
\[
{\rm start}(x),P_\#(x,y) \, \rightarrow \, {\rm sync}(y,y).
\]
That is, the first element after the special symbol $\#$ (which denotes the beginning of a word over $\{a,b\}$)
is synchronized with itself.

\item For each $1 \leq i \leq n$, a synchronization rule:
\[
{\rm sync}(x,y),P_{w_i}(x,z),P_{w'_i}(y,u) \, \rightarrow \, {\rm
sync}(z,u).
\]
Here, $P_{w}(x,y)$, for $w = a_1 \dots a_t \in \{a,b\}^+$, denotes $P_{a_1}(x,x_1),\dots,$ $P_{a_t}(x_{t-1},y)$, where the $x_i$'s are fresh variables.
Roughly, if $(x,y)$ is synchronized and the element $z$ (resp., $u$) is reachable from $x$ (resp., $y$) by word $w_i$ (resp., $w'_i$), then $(z,u)$ is also synchronized.

\item For each $1 \leq i \leq n$, a finalization rule:
\begin{multline*}
{\rm start}(x),P_a(y,z),P_a(z,u),P_*(u,v),{\rm end}(v),\\
{\rm sync}(y_1,y_2),P_{w_i}(y_1,y),P_{w'_i}(y_2,y)
\, \rightarrow \, \psi,
\end{multline*}
where
$\psi$ is the conjunction of atoms:
\begin{multline*}
P_\#(x,y),P_\#(x,z),P_\#(x,u),P_*(y,v),P_*(z,v),\\
P_b(z,y),P_b(u,z),P_a(u,y),P_b(y,u), \\
{\rm sync}(y,y),{\rm sync}(z,z),{\rm sync}(y,z),{\rm sync}(z,y),\\
{\rm sync}(y,u),{\rm sync}(u,y),{\rm sync}(z,u),{\rm sync}(u,z).
\end{multline*}

This tgd enforces $\chase{q'}{\Sigma}$ to contain
a ``copy'' of $q$ whenever $q'$ encodes a solution for the PCP
instance.
\end{enumerate}

We first show that if the PCP instance has a solution given by the
nonempty sequence $i_1 \dots i_m$, with $1 \leq i_1,\dots,i_m \leq n$,
then there exists an acylic CQ $q'$ whose underlying graph is a
directed path
such that $q \equiv_\Sigma q'$. Let us assume that
$w_{i_1} \dots w_{i_m} = a_1 \dots a_t$, where each $a_i \in \{a,b\}$.
It is not hard to prove that $q \equiv_\Sigma q'$, where
$q'$ is as follows:
\begin{center}
\begin{picture}(0,0)%
\includegraphics{acq-body.pstex}%
\end{picture}%
\setlength{\unitlength}{2052sp}%
\begingroup\makeatletter\ifx\SetFigFont\undefined%
\gdef\SetFigFont#1#2#3#4#5{%
  \reset@font\fontsize{#1}{#2pt}%
  \fontfamily{#3}\fontseries{#4}\fontshape{#5}%
  \selectfont}%
\fi\endgroup%
\begin{picture}(7373,570)(2911,-1846)
\put(9601,-1486){\makebox(0,0)[lb]{\smash{{\SetFigFont{5}{6.0}{\familydefault}{\mddefault}{\updefault}{\color[rgb]{0,0,0}$P_{*}$}%
}}}}
\put(2926,-1786){\makebox(0,0)[lb]{\smash{{\SetFigFont{5}{6.0}{\familydefault}{\mddefault}{\updefault}{\color[rgb]{0,0,0}${\rm start}$}%
}}}}
\put(3451,-1486){\makebox(0,0)[lb]{\smash{{\SetFigFont{5}{6.0}{\familydefault}{\mddefault}{\updefault}{\color[rgb]{0,0,0}$P_{\#}$}%
}}}}
\put(4351,-1486){\makebox(0,0)[lb]{\smash{{\SetFigFont{5}{6.0}{\familydefault}{\mddefault}{\updefault}{\color[rgb]{0,0,0}$P_{a_1}$}%
}}}}
\put(3151,-1411){\makebox(0,0)[lb]{\smash{{\SetFigFont{5}{6.0}{\familydefault}{\mddefault}{\updefault}{\color[rgb]{0,0,0}$x'$}%
}}}}
\put(10051,-1786){\makebox(0,0)[lb]{\smash{{\SetFigFont{5}{6.0}{\familydefault}{\mddefault}{\updefault}{\color[rgb]{0,0,0}${\rm end}$}%
}}}}
\put(7501,-1411){\makebox(0,0)[lb]{\smash{{\SetFigFont{5}{6.0}{\familydefault}{\mddefault}{\updefault}{\color[rgb]{0,0,0}$y'$}%
}}}}
\put(8401,-1411){\makebox(0,0)[lb]{\smash{{\SetFigFont{5}{6.0}{\familydefault}{\mddefault}{\updefault}{\color[rgb]{0,0,0}$z'$}%
}}}}
\put(9301,-1411){\makebox(0,0)[lb]{\smash{{\SetFigFont{5}{6.0}{\familydefault}{\mddefault}{\updefault}{\color[rgb]{0,0,0}$u'$}%
}}}}
\put(10201,-1411){\makebox(0,0)[lb]{\smash{{\SetFigFont{5}{6.0}{\familydefault}{\mddefault}{\updefault}{\color[rgb]{0,0,0}$v'$}%
}}}}
\put(6826,-1486){\makebox(0,0)[lb]{\smash{{\SetFigFont{5}{6.0}{\familydefault}{\mddefault}{\updefault}{\color[rgb]{0,0,0}$P_{a_t}$}%
}}}}
\put(5476,-1561){\makebox(0,0)[lb]{\smash{{\SetFigFont{9}{10.8}{\familydefault}{\mddefault}{\updefault}{\color[rgb]{0,0,0}$\dots$}%
}}}}
\put(7801,-1486){\makebox(0,0)[lb]{\smash{{\SetFigFont{5}{6.0}{\familydefault}{\mddefault}{\updefault}{\color[rgb]{0,0,0}$P_{a}$}%
}}}}
\put(8701,-1486){\makebox(0,0)[lb]{\smash{{\SetFigFont{5}{6.0}{\familydefault}{\mddefault}{\updefault}{\color[rgb]{0,0,0}$P_{a}$}%
}}}}
\end{picture}%

\end{center}
Here, again, $x',y',z',u',v'$ denote the names of the respective
variables of $q'$.  All nodes in the above path are different. The main reason why $q \equiv_\Sigma q'$
holds is because the fact $w$ is a solution implies that there are elements $y_1$ and $y_2$ such that
${\rm sync}(y_1,y_2)$, $P_{w_1}(y_1,y)$ and $P_{w'_i}(y_2,y)$ hold in $\chase{q'}{\Sigma}$. Thus, the finalization rule is fired. This creates a copy of $q$ in $\chase{q'}{\Sigma}$, which allows $q$ to be
homomorphically mapped to $\chase{q'}{\Sigma}$.

Now we prove that if there exists an acyclic CQ $q'$ such that $q
\equiv_\Sigma q'$ and the underlying graph of $q'$ is a directed path,
then the PCP instance has a solution.
Since $q
\equiv_\Sigma q'$, Lemma
\ref{lemma:cq-equiv-tgds} tells us that $\chase{q}{\Sigma}
\equiv \chase{q'}{\Sigma}$ are homomorphically equivalent. But then
$\chase{q'}{\Sigma}$ must contain at least one variable labeled ${\rm
  start}$ and one variable labeled ${\rm end}$. The first variable cannot have incoming edges (otherwise,
$\chase{q'}{\Sigma}$ would not homomorphically map to
$\chase{q}{\Sigma}$), while the second one cannot have outcoming edges
(for the same reason). Thus, it is the first variable $x'$ of
$q'$ that is labeled ${\rm start}$ and the last one $v'$ that is
labeled ${\rm end}$. Further, all edges
reaching $v'$ in $q'$ must be labeled
$P_*$ (otherwise $q'$ does not homomorphically map to $q$).
Thus, this is the label of the last edge of $q'$ that goes
from variable $u'$ to $v'$. Analogously, the edge that leaves $x'$ in
$q'$ is labeled $P_\#$. Further, any other edge in $q'$ is
labeled $P_a$, $P_b$, or ${\rm sync}$.

Notice now that $v'$ must have an incoming edge labeled $P_*$ in
$\chase{q'}{\Sigma}$ from
some node $u''$ that has an outgoing edge with label $P_a$ (since $q$ homomorphically maps
to $\chase{q'}{\Sigma}$).
By definition of $\Sigma$, this could
only have happened if the finalization rule is fired. In particular,
$u'$ is preceded by node
$z'$, which in turn is preceded by $y'$,
and there are elements $y'_1$ and $y'_2$ such that
${\rm sync}(y'_1,y'_2)$, $P_{w_1}(y'_1,y')$ and $P_{w'_i}(y'_2,y')$ hold in $\chase{q'}{\Sigma}$.
 In fact, the
unique path from $y'_1$ (resp., $y'_2$) to $y'$ in $q'$ is labeled $w_i$
(resp., $w'_i$). This means that the atom ${\rm sync}(y'_1,y'_2)$ was
not one of the edges of $q'$, but must have been produced during the
chase by firing the initialization or the synchronization rules, and so on.
This process must finish in the second element $x^*$ of
$q'$. (Recall that ${\rm sync}(x^*,x^*)$ belongs to $\chase{q'}{\Sigma}$
due to the first rule of $\Sigma$).
We conclude that our PCP instance has a solution.
\end{proof}

Theorem \ref{theo:full-undec} rules out any class that captures the class of full tgds, e.g., weakly-guarded, weakly-acyclic and weakly-sticky sets of tgds.
The question that comes up is whether the non-weak versions of the
above classes, namely guarded, non-recursive and sticky sets of tgds,
ensure the decidability of $\SA$, and what is the complexity of the problem. This is the subject of the next two
sections.

\section{Acyclicity-Preserving Chase}\label{sec:apc}

We propose a semantic criterion, the so-called {\em acyclicity-preserving chase}, that ensures the decidability of $\SA(\class{C})$ whenever the problem $\Cont(\class{C})$ is decidable. This criterion guarantees that, starting from an acyclic instance, it is not possible to destroy its acyclicity during the construction of the chase.
We then proceed to show that the class of guarded sets of tgds has acyclicity-preserving chase, which immediately implies the decidability of $\SA(\guarded)$, and we pinpoint the exact complexity of the latter problem. Notice that non-recursiveness and stickiness do not enjoy this property, even in the restrictive setting where only unary and binary predicates can be used; more details are given in the next section. The formal definition of our semantic criterion follows:

\begin{definition}\label{def:apc}(\textbf{Acyclicity-preserving Chase})
We say that a class $\class{C}$ of sets of tgds has {\em acyclicity-preserving chase} if, for every acyclic CQ $q$, set $\dep \in \class{C}$, and chase sequence for $q$ under $\dep$, the result of such a chase sequence is acyclic. \hfill\markfull
\end{definition}

We can then prove the following small query property:

\begin{proposition}\label{prop:apc-small-query-property}
Let $\Sigma$ be a finite set of tgds that belongs to a class that has acyclicity-preserving chase, and $q$ a CQ. If $q$ is semantically acyclic under $\dep$, then there exists an acyclic CQ $q'$, where $|q'| \leq 2 \cdot |q|$, such that $q \equiv_{\dep} q'$.
\end{proposition}

\begin{figure}[t]
 \epsfclipon
  \centerline
  {\hbox{
  \leavevmode
  \epsffile{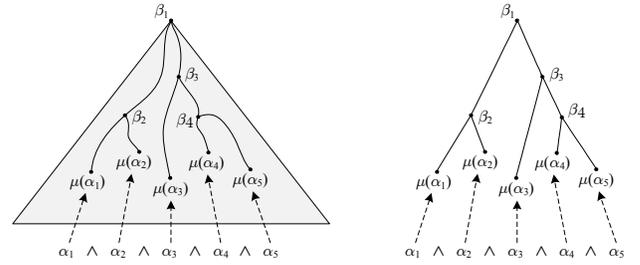}
  }} \epsfclipoff \caption{The compact acyclic query.}
  \label{fig:apc}
  \vspace{-2mm}
\end{figure}

The proof of the above result relies on the following technical lemma, established in~\cite{CaGK13} (using slightly different terminology), that will also be used later in our investigation:

\begin{lemma}\label{lem:from-cq-to-acq}
Let $q({\bar x})$ be a CQ, $I$ an acyclic instance, and ${\bar c}$ a tuple of distinct constants occurring in $I$ such that $q({\bar c})$ holds in $I$. There exists an acyclic CQ $q'({\bar x})$, where $q' \subseteq q$ and $|q'| \leq 2 \cdot |q|$, such that $q'({\bar c})$ holds in $I$.
\end{lemma}

For the sake of completeness, we would like to recall the idea of the construction underlying Lemma~\ref{lem:from-cq-to-acq}, which is illustrated in Figure~\ref{fig:apc}. Assuming that $\alpha_1, \ldots, \alpha_5$ are the atoms of $q$, there exists a homomorphism $\mu$ that maps $\alpha_1 \wedge \ldots \wedge \alpha_5$ to the join tree $T$ of the acyclic instance $I$ (the shaded tree in Figure~\ref{fig:apc}). Consider now the subtree $T_q$ of $T$ consisting of all the nodes in the image of the query and their ancestors. From $T_q$ we extract the smaller tree $F$ also depicted in Figure~\ref{fig:apc}; $F = (V,E)$ is obtained as follows:
\begin{enumerate}\itemsep-\parsep
\item $V$ consists of all the root and leaf nodes of $T_q$, and all the inner nodes of $T_q$ with at least two children; and

\item For every $v,u \in V$, $(v,u) \in E$ iff $u$ is a descendant of $v$ in $T_q$, and the only nodes of $V$ that occur on the unique shortest path from $v$ to $u$ in $T_q$ are $v$ and $u$.
\end{enumerate}
It is easy to verify that $F$ is a join tree, and has at most $2 \cdot |q|$ nodes. The acyclic conjunctive query $q'$ is defined as the conjunction of all atoms occurring in $F$.

Notice that a result similar to Lemma~\ref{lem:from-cq-to-acq} is implicit in~\cite{BaLR14}, where the problem of approximating conjunctive queries is investigated. However, from the results of~\cite{BaLR14}, we can only conclude the existence of an exponentially sized acyclic CQ in the arity of the underlying schema, while Lemma~\ref{lem:from-cq-to-acq} establishes the existence of an acyclic query of linear size. This is decisive for our later complexity analysis.
Having the above lemma in place, it is not difficult to establish Proposition~\ref{prop:apc-small-query-property}.

\medskip

\begin{proof}[of Proposition~\ref{prop:apc-small-query-property}]
Since, by hypothesis, $q$ is semantically acyclic under $\dep$, there exists an acyclic CQ $q''({\bar x})$ such that $q \equiv_{\dep} q''$. By Lemma~\ref{lemma:cq-equiv-tgds}, $c({\bar x})$ belongs to the evaluation of $q$ over $\chase{q''}{\dep}$. Recall that $\dep$ belongs to a class that has acyclicity-preserving chase, which implies that $\chase{q''}{\dep}$ is acyclic. Hence, by Lemma~\ref{lem:from-cq-to-acq}, there exists an acyclic CQ $q'$, where $q' \subseteq q$ and $|q'| \leq 2 \cdot |q|$, such that $c({\bar x})$ belongs to the evaluation of $q'$ over $\chase{q''}{\dep}$. By Lemma~\ref{lemma:cq-equiv-tgds}, $q'' \subseteq_{\dep} q'$, and therefore $q \subseteq_{\dep} q'$. We conclude that $q \equiv_{\dep} q'$, and the claim follows.
\end{proof}

It is clear that Proposition~\ref{prop:apc-small-query-property} provides a decision procedure for $\SA(\class{C})$ whenever $\class{C}$ has acyclicity-preserving chase and $\Cont(\class{C})$ is decidable. Given a CQ $q$, and a finite set $\dep \in \class{C}$:
\begin{enumerate}
\item Guess an acyclic CQ $q'$ of size at most $2 \cdot |q|$; and

\item Verify that $q \subseteq_{\dep} q'$ and $q' \subseteq_{\dep} q$.
\end{enumerate}

The next result follows:

\begin{theorem}\label{theo:apc}
Consider a class $\class{C}$ of sets of tgds that has acyclicity-preserving chase. If the problem $\Cont(\class{C})$ is decidable, then $\SA(\class{C})$ is also decidable.
\end{theorem}

\subsection{Guardedness}

We proceed to show that $\SA(\guarded)$ is decidable and has the same complexity as CQ containment under guarded tgds:

\begin{theorem}\label{the:sa-guarded}
$\SA(\guarded)$ is complete for \twoexptime. It becomes
\EXP-complete if the arity of the schema is fixed, and \NP-complete if the schema is fixed.
\end{theorem}

The rest of this section is devoted to establish Theorem~\ref{the:sa-guarded}.

\subsubsection*{\textbf{Decidability and Upper Bounds}}

We first show that:

\begin{proposition}\label{pro:guarded-tgds-apc}
$\guarded$ has acyclicity-preserving chase.
\end{proposition}

The above result, combined with Theorem~\ref{theo:apc}, implies the decidability of $\SA(\guarded)$. However, this does not say anything about the complexity of the problem. With the aim of pinpointing the exact complexity of
$\SA(\guarded)$, we proceed to analyze the complexity of the decision procedure underlying Theorem~\ref{theo:apc}.
Recall that, given a CQ $q$, and a finite set $\dep \in \guarded$, we guess an acyclic CQ $q'$ such that $|q'| \leq 2 \cdot |q|$, and verify that $q \equiv_{\dep} q'$. It is clear that this algorithm runs in non-deterministic polynomial time with a call to a $\C$ oracle, where $\C$ is a complexity class powerful enough for solving $\Cont(\guarded)$. Thus, Proposition~\ref{prop:qa-guarded} implies that $\SA(\guarded)$ is in \twoexptime, in \EXP~if the arity of the schema is fixed, and in \NP~if the schema is fixed.
One may ask why for a fixed schema the obtained upper bound is \NP~and not $\Sigma_{2}^{P}$. Observe that the oracle is called only once in order to solve $\Cont(\guarded)$, and since $\Cont(\guarded)$ is already in \NP~when the schema is fixed, it is not really needed in this case.

\subsubsection*{\textbf{Lower Bounds}}

Let us now show that the above upper bounds are optimal. By Proposition~\ref{prop:gral-reduction}, $\RCont(\class{\guarded})$ can be reduced in constant time to $\SA(\guarded)$. Thus, to obtain the desired lower bounds, it suffices to reduce in polynomial time $\Cont(\guarded)$ to $\RCont(\class{\guarded})$.
Interestingly, the lower bounds given in Section~\ref{sec:prelim} for $\Cont(\guarded)$ hold even if we focus on Boolean CQs and the left-hand side query is acyclic. In fact, this is true, not only for guarded, but also for non-recursive and sticky sets of tgds.
Let $\ABCont(\class{C})$ be the following problem: Given an acyclic Boolean CQ $q$, a Boolean CQ $q'$, and a finite set $\dep \in \class{C}$ of tgds, is it the case $q \subseteq_{\dep} q'$?

From the above discussion, to establish the desired lower bounds for guarded sets of tgds (and also for the other classes of tgds considered in this work), it suffices to reduce in polynomial time $\ABCont$ to $\RCont$. To this end, we introduce the so-called connecting operator, which provides a generic reduction from $\ABCont$ to $\RCont$.

\medskip
\noindent \paragraph{Connecting operator.}
Consider an acyclic Boolean CQ $q$, a Boolean CQ $q'$, and a finite set $\dep$ of tgds. We assume that both $q,q'$ are of the form $\exists \bar y \big(R_1(\bar v_1) \wedge \dots \wedge R_m(\bar v_m)\big)$.
The application of the {\em connecting operator} on $(q,q',\dep)$ returns the triple $(\mathsf{c}(q),\mathsf{c}(q'),\mathsf{c}(\dep))$, where
\begin{itemize}
\item $\mathsf{c}(q)$ is the query
    \[
    \exists \bar y \exists w \big(R_{1}^{\star}(\bar v_1,w) \wedge \dots \wedge R_{m}^{\star}(\bar v_m,w) \wedge \mi{aux}(w,w)\big),
    \]
    where $w$ is a new variable not in $q$, each $R_{i}^{\star}$ is a new predicate, and also $\mi{aux}$ is a new binary predicate;

\item $\mathsf{c}(q')$ is the query
    \begin{multline*}
    \exists \bar y \exists w \exists u \exists v \big(R_{1}^{\star}(\bar v_1,w) \wedge \dots \wedge R_{m}^{\star}(\bar v_m,w) \, \wedge\\
    \hspace{23mm} \mi{aux}(w,u)\ \wedge\ \mi{aux}(u,v)\ \wedge\ \mi{aux}(v,w)\big),
    \end{multline*}
    where $w,u,v$ are new variables not in $q$; and

\item Finally, $\mathsf{c}(\dep) = \{\mathsf{c}(\tau) \mid \tau \in \dep\}$, where for a tgd $\tau = \phi(\bar x,\bar y) \rightarrow \exists \bar z \psi(\bar x,\bar z)$, $\mathsf{c}(\tau)$ is the tgd
    \[
    \phi^\star(\bar x,\bar y,w) \rightarrow \exists \bar z \psi^\star(\bar x,\bar z,w),
    \]
    with $\phi^\star(\bar x,\bar y,w), \psi^\star(\bar x,\bar z,w)$ be the conjunctions obtained from $\phi(\bar x,\bar y), \psi(\bar x,\bar z)$, respectively, by replacing each atom $R(x_1,\ldots,x_n)$ with $R^\star(x_1,\ldots,x_n,w)$, where $w$ is a new variable not occurring in $\tau$.
\end{itemize}
This concludes the definition of the connecting operator. A class $\class{C}$ of sets of tgds is {\em closed under connecting} if, for every set $\dep \in \class{C}$, $\mathsf{c}(\dep) \in \class{C}$.
It is easy to verify that $\mathsf{c}(q)$ remains acyclic and is connected, $\mathsf{c}(q')$ is connected and not semantically acyclic under $\mathsf{c}(\dep)$, and $\mathsf{c}(\dep)$ is a set of body-connected tgds.
It can be also shown that $q \subseteq_{\dep} q'$ iff $\mathsf{c}(q) \subseteq_{\mathsf{c}(\dep)} \mathsf{c}(q')$.

From the above discussion, it is clear that the connecting operator provides a generic polynomial time reduction from $\ABCont(\class{C})$ to $\RCont(\class{C})$, for every class $\class{C}$ of sets of tgds that is closed under connecting. Then:

\begin{proposition}\label{pro:generic-lower-bound}
Let $\class{C}$ be a class of sets of tgds that is closed under connecting such that $\ABCont(\class{C})$ is hard for a complexity class $\C$ that is closed under polynomial time reductions. Then, $\SA(\class{C})$ is also $\C$-hard.
\end{proposition}

\paragraph{Back to guardedness.} It is easy to verify that the class of guarded sets of tgds is closed under connecting. Thus, the lower bounds for $\SA(\guarded)$ stated in Theorem~\ref{the:sa-guarded} follow from Propositions~\ref{prop:qa-guarded} and~\ref{pro:generic-lower-bound}. Note that, although Proposition~\ref{prop:qa-guarded} refers to $\Cont(\guarded)$, the lower bounds hold for $\ABCont(\guarded)$; this is implicit in~\cite{CaGK13}.

As said in Section~\ref{sec:prelim}, a key subclass of guarded sets of tgds is the class of linear tgds, i.e., tgds whose body consists of a single atom, which in turn subsume the well-known class of inclusion dependencies. By exploiting the non-deterministic procedure employed for $\SA(\guarded)$, and the fact that both linear tgds and inclusion dependencies are closed under connecting, we can show that:

\begin{theorem}\label{the:sa-linear}
$\SA(\class{C})$, for $\class{C} \in \{\linear,\id\}$, is complete for \PSPACE. It becomes \NP-complete if the arity of the schema is fixed.
\end{theorem}

\section{UCQ Rewritability}
\label{sec:ucq}

Even though the acyclicity-preserving chase criterion was very useful for solving $\SA(\guarded)$, it is of little use for non-recursive and sticky sets of tgds. As we show in the next example, neither $\class{NR}$ nor $\class{S}$ have acyclicity-preserving chase:

\begin{example}\label{exa:nr-sticky-violate-apc}
Consider the acyclic CQ and the tgd
\[
q = \exists \bar x \big(P(x_1) \wedge \ldots \wedge P(x_n)\big) \quad \tau = P(x),P(y) \ra R(x,y),
\]
where $\{\tau\}$ is both non-recursive and sticky, but not guarded.
In $\chase{q}{\{\tau\}}$ the predicate $R$ holds all the possible pairs that can be formed using the terms $x_1,\ldots,x_n$. Thus, in the Gaifman graph of $\chase{q}{\{\tau\}}$ we have an $n$-clique, which means that is highly cyclic.
Notice that our example illustrates that also other favorable properties of the CQ are destroyed after chasing with non-recursive and sticky sets of tgds, namely bounded (hyper)tree width.\footnote{Notice that guarded sets of tgds over predicates of bounded arity preserve the bounded hyper(tree) width of the query.} \hfill\markfull
\end{example}

In view of the fact that the methods devised in Section~\ref{sec:apc} cannot be used for non-recursive and sticky sets of tgds, new techniques must be developed. Interestingly, $\class{NR}$ and $\class{S}$ share an important property, which turned out to be very useful for semantic acyclicity: {\em UCQ rewritability}.
Recall that a {\em union of conjunctive queries (UCQ)} is an expression of the form $Q(\bar x) = \bigvee_{1 \leq i \leq n} q_i(\bar x)$, where each $q_i$ is a CQ over the same schema $\sigma$. The evaluation of $Q$ over an instance $I$, denoted $Q(I)$, is defined as $\bigcup_{1 \leq i \leq n} q_i(I)$.
The formal definition of UCQ rewritability follows:

\begin{definition}\label{def:ucq-rewritability}(\textbf{UCQ Rewritability})
A class $\class{C}$ of sets of tgds is {\em UCQ rewritable} if, for every CQ $q$, and $\dep \in \class{C}$, we can construct a UCQ $Q$ such that: For every CQ $q'(\bar x)$, $q' \subseteq_{\dep} q$ iff $c(\bar x) \in Q(D_{q'})$, with $D_{q'}$ be the database obtained from $q'$ after replacing each variable $x$ with $c(x)$. \hfill\markfull
\end{definition}

In other words, UCQ rewritability suggests that query containment can be reduced to the problem of UCQ evaluation. It is important to say that this reduction depends only on the right-hand side CQ and the set of tgds, but not on the left-hand side query. This is crucial for establishing the desirable small query property whenever we focus on sets of tgds that belong to a UCQ rewritable class.
At this point, let us clarify that the class of guarded sets of tgds is not UCQ rewritable, which justifies our choice of a different semantic property, that is, acyclicity-preserving chase, for its study. 

%
%

Let us now show the desirable small query property. For each UCQ rewritable class $\class{C}$ of sets of tgds, there exists a computable function $f_{\class{C}}(\cdot,\cdot)$ from the set of pairs consisting of a CQ and a set of tgds in $\class{C}$ to positive integers such that the following holds: For every CQ $q$, set $\dep \in \class{C}$, and UCQ rewriting $Q$ of $q$ and $\dep$, the {\em height} of $Q$, that is, the maximal size of its disjuncts, is at most $f_{\class{C}}(q,\dep)$. The existence of the function $f_{\class{C}}$ follows by the definition of UCQ rewritability. Then, we show the following:

\begin{proposition}\label{prop:ucq-rewritability-small-query-property}
Let $\class{C}$ be a UCQ rewritable class, $\dep \in \class{C}$ a finite set of tgds, and $q$ a CQ. If $q$ is semantically acyclic under $\dep$, then there exists an acyclic CQ $q'$, where $|q'| \leq 2 \cdot f_{\class{C}}(q,\dep)$, such that $q \equiv_{\dep} q'$.
\end{proposition}

\begin{proof}
  Since $q$ is semantically acyclic under $\dep$, there exists an
  acyclic CQ $q''({\bar x})$ such that $q \equiv_{\dep} q''$. As $\class{C}$ is UCQ rewritable, there exists a UCQ $Q$ such that
  $c(\bar x) \in Q(D_{q''})$, which implies that there exists a CQ
  $q_r$ (one of the disjuncts of $Q$) such that $c(\bar x) \in
  q_r(D_{q''})$. Clearly, $|q_r| \leq f_{\class{C}}(q,\dep)$.
  But $D_{q''}$ is acyclic, and thus Lemma~\ref{lem:from-cq-to-acq}
  implies the existence of an acyclic CQ $q'$, where $q' \subseteq
  q_r$ and $|q'| \leq 2 \cdot f_{\class{C}}(q,\dep)$, such that
  $c({\bar x}) \in q'(D_{q''})$. The latter implies that $q''
  \subseteq q'$. By hypothesis, $q \subseteq_{\dep} q''$, and hence $q
  \subseteq_{\dep} q'$. For the other direction, we first show that
  $q_r \subseteq_{\dep} q$ (otherwise, $Q$ is not a UCQ rewriting).
  Since $q' \subseteq q_r$, we get that $q' \subseteq_{\dep} q$. We
  conclude that $q \equiv_{\dep} q'$, and the claim follows.
\end{proof}

It is clear that Proposition~\ref{prop:ucq-rewritability-small-query-property} provides a decision procedure for $\SA(\class{C})$ whenever $\class{C}$ is UCQ rewritable, and $\Cont(\class{C})$ is decidable. Given a CQ $q$, and a finite set $\dep \in \class{C}$:
\begin{enumerate}
\item Guess an acyclic CQ $q'$ of size at most $2 \cdot f_{\class{C}}(q,\dep)$; and

\item Verify that $q \subseteq_{\dep} q'$ and $q' \subseteq_{\dep} q$.
\end{enumerate}

The next result follows:

\begin{theorem}\label{theo:ucq-rewritability}
Consider a class $\class{C}$ of sets of tgds that is UCQ rewritable. If the problem $\Cont(\class{C})$ is decidable, then $\SA(\class{C})$ is also decidable.
\end{theorem}

\subsection{Non-Recursiveness}

As already said, the key property of $\nr$ that we are going to exploit for solving $\SA(\nr)$ is UCQ rewritability.
For a CQ $q$ and a set $\dep$ of tgds, let $p_{q,\dep}$ and $a_{q,\dep}$ be the number of predicates in $q$ and $\dep$, and the maximum arity over all those predicates, respectively.
The next result is implicit in~\cite{GoOP14}:\footnote{The work~\cite{GoOP14} does not consider $\nr$. However, the rewriting algorithm in that paper works also for non-recursive sets of tgds.}

\begin{proposition}\label{pro:nr-ucq-rewritable}
$\nr$ is UCQ rewritable. Furthermore, $f_{\nr}(q,\dep) = p_{q,\dep} \cdot (a_{q,\dep} \cdot |q| + 1)^{a_{q,\dep}}$.
\end{proposition}

The above result, combined with Theorem~\ref{theo:ucq-rewritability}, implies the decidability of $\SA(\nr)$. For the exact complexity of the problem, we simply need to analyze the complexity of the non-deterministic algorithm underlying Theorem~\ref{theo:ucq-rewritability}.
Observe that when the arity of the schema is fixed the function $f_{\nr}$ is polynomial, and therefore Proposition~\ref{pro:nr-ucq-rewritable} guarantees the existence of a polynomially sized acyclic CQ. In this case, by exploiting Proposition~\ref{prop:qa-nr}, it is easy to show that $\SA(\nr)$ is in \NEXP, and in \NP~if the schema is fixed.
However, things are a bit cryptic when the arity of the schema is not fixed. In this case, $f_{\nr}$ is exponential, and thus we have to guess an acyclic CQ of exponential size. But now the fact that $\Cont(\nr)$ is in \NEXP~(by Proposition~\ref{prop:qa-nr}) alone is not enough to conclude that $\SA(\nr)$ is also in \NEXP.
We need to understand better the complexity of the query containment algorithm for $\nr$.

Recall that given two CQs $q(\bar x), q'(\bar x)$, and a finite set $\dep \in \nr$, by Lemma~\ref{lemma:cq-equiv-tgds}, $q \subseteq_{\dep} q'$ iff $c(\bar x) \in q'(\chase{q}{\dep})$. By exploiting non-recursiveness, it can be shown that if $c(\bar x) \in q'(\chase{q}{\dep})$, then there exists a chase sequence
\[
q = I_0 \xrightarrow{\tau_0,\bar c_0} I_1 \xrightarrow{\tau_1,\bar c_1} I_2 \ldots I_{n-1} \xrightarrow{\tau_{n-1},\bar c_{n-1}} I_n
\]
of $q$ and $\dep$, where $n = |q'| \cdot (b_\dep)^{\O(p_{q',\dep})}$, with $b_\dep$ be the maximum number of atoms in the body of a tgd of $\dep$, such that $c(\bar x) \in q'(I_n)$.
The query containment algorithm for $\nr$ simply guesses such a chase sequence of $q$ and $\dep$, and checks whether $c(\bar x) \in q'(I_n)$. Since $n$ is exponential, this algorithm runs in non-deterministic exponential time.
Now, recall that for $\SA(\nr)$ we need to perform two containment checks where either the left-hand side or the right-hand side query is of exponential size. But in both cases the containment algorithm for $\nr$ runs in non-deterministic exponential time, and hence $\SA(\nr)$ is in \NEXP.
The lower bounds are inherited from $\ABCont(\nr)$ since $\nr$ is closed under connecting (see Proposition~\ref{pro:generic-lower-bound}). Then:

\begin{theorem}\label{the:sa-nr}
$\SA(\nr)$ is complete for \NEXP, even if the arity of the schema is fixed. It becomes \NP-complete if the schema is fixed.
\end{theorem}

\subsection{Stickiness}

We now focus on sticky sets of tgds. As for $\nr$,
the key property of $\sticky$ that we are going to use is UCQ rewritability. The next result has been explicitly shown in~\cite{GoOP14}:

\begin{proposition}\label{pro:sticky-ucq-rewritable}
$\sticky$ is UCQ rewritable. Furthermore, $f_{\sticky}(q,\dep) = p_{q,\dep} \cdot (a_{q,\dep} \cdot |q| + 1)^{a_{q,\dep}}$.
\end{proposition}

The above result, combined with Theorem~\ref{theo:ucq-rewritability}, implies the decidability of $\SA(\sticky)$. Moreover, Proposition~\ref{pro:sticky-ucq-rewritable} allows us to establish an optimal upper bound when the arity of the schema is fixed since in this case the function $f_{\sticky}$ is polynomial. In fact, we show that $\SA(\sticky)$ is \NP-complete when the arity of the schema is fixed. The \NP-hardness is inherited from $\ABCont(\sticky)$ since $\sticky$ is closed under connecting (see Proposition~\ref{pro:generic-lower-bound}).
Now, when the arity of the schema is not fixed the picture is still foggy. In this case, the function $f_{\sticky}$ is exponential, and thus by following the usual guess and check approach we get that $\SA(\sticky)$ is in \NEXP, while Proposition~\ref{pro:generic-lower-bound} implies an \EXP~lower bound.
To sum up, our generic machinery based on UCQ rewritability shows that:

\begin{theorem}\label{the:sa-sticky}
$\SA(\sticky)$ is in \NEXP~and hard for \EXP. It becomes \NP-complete if the arity is fixed.
\end{theorem}

An interesting question that comes up is whether for sticky sets of
tgds a stronger small query property than
Proposition~\ref{prop:ucq-rewritability-small-query-property} can be
established, which guarantees the existence of a polynomially sized
equivalent acyclic CQ. It is clear that such a result would allow us to
establish an \EXP~upper bound for $\SA(\sticky)$.
At this point, one might be tempted to think that this can be achieved by showing that the function $f_{\sticky}$ is actually polynomial even if the arity of the schema is not fixed.
The next example shows that this is not the case. We can construct a sticky set $\dep$ of tgds and a CQ $q$ such that, for every UCQ rewriting $Q$ for $q$ and $\dep$, the height of $Q$ is exponential in the arity.

\begin{example}\label{exa:function-fs-lower-bound}
Let $\dep$ be the sticky set of tgds given below; we write ${\bar x}_{i}^{j}$ for the tuple of variables $x_i,x_{i+1},\ldots,x_j$:
\begin{align*}
&\big\{P_i({\bar x}_{1}^{i-1},Z,{\bar x}_{i+1}^{n},Z,O),
P_i({\bar x}_{1}^{i-1},O,{\bar x}_{i+1}^{n},Z,O) \ra\\
&\hspace{29mm} P_{i-1}({\bar x}_{1}^{i-1},Z,{\bar x}_{i+1}^{n},Z,O)\big\}_{i \in \{1,\ldots,n\}}.
\end{align*}
Consider also the Boolean CQ
\[
P_0(0,\ldots,0,0,1).
\]
It can be shown that, for every UCQ rewriting $Q$ for $q$ and $\dep$, the disjunct of $Q$ that mentions only the predicate $P_n$ contains exactly $2^n$ atoms. Therefore, there is no UCQ rewriting for $q$ and $\dep$ of polynomial height, which in turn implies that $f_{\sticky}$ cannot be polynomial in the arity of the schema. \hfill\markfull
\end{example}

The above discussion reveals the need to identify a more refined property of stickiness than UCQ rewritability, which will allow us to close the complexity of $\SA(\sticky)$ when the arity is not fixed. This is left as an interesting open problem. 
\section{Semantic Acyclicity with EGDs}\label{sec:semac-egds}

Up to now, we have considered classes of constraints that are based on tgds. However, semantic acyclicity can be naturally defined for classes of egds. Hence, one may wonder whether the techniques developed in the previous sections can be applied for egd-based classes of constraints. Unfortunately, the situation changes dramatically even for the simplest subclass of egds, i.e., keys.

\subsection{Peculiarity of Keys}

We show that the techniques developed in the previous sections for tgds cannot be applied for showing the decidability of semantic acyclicity under keys, and thus under egds. Although the notions of acyclicity-preserving chase (Definition~\ref{def:apc}) and UCQ rewritability (Definition~\ref{def:ucq-rewritability}) can be naturally defined for egds, are of little use even if we focus on keys.

\medskip
\noindent \paragraph{Acyclicity-preserving chase.} It is easy to show via a simple example that keys over binary and ternary predicates do not enjoy the acyclicity-preserving chase property:

\begin{example}\label{exa:keys-violate-acyclicity}
Let $q$ be the acyclic query
\[
R(x,y) \wedge S(x,y,z) \wedge S(x,z,w) \wedge S(x,w,v) \wedge R(x,v).
\]
After applying on $q$ the key $R(x,y),R(x,z) \ra y = z$, which simply states that the first attribute of the binary predicate $R$ is the key, we obtain the query
\[
R(x,y) \wedge S(x,y,z) \wedge S(x,z,w) \wedge S(x,w,y),
\]
which is clearly cyclic. \hfill\markfull
\end{example}

With the aim of emphasizing the peculiarity of keys, we give a more involved example, which shows that a tree-like query can be transformed via two keys into a highly cyclic query that contains a grid. Interestingly, this shows that also other desirable properties, in particular bounded (hyper)tree width, are destroyed when we chase a query using keys.

\begin{figure}[t]
 \epsfclipon
  \centerline
  {\hbox{
  \leavevmode
  \epsffile{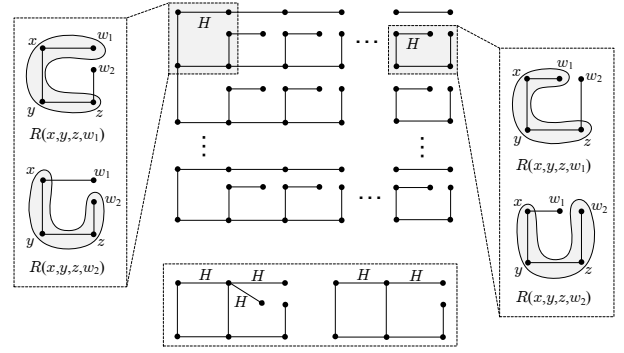}
  }} \epsfclipoff \caption{From a ``tree'' to a grid via key dependencies.}
  \label{fig:tree-to-grid}
  \vspace{-2mm}
\end{figure}

\begin{example}\label{exa:tree-to-grid}
Consider the CQ $q$ depicted in Figure~\ref{fig:tree-to-grid} (ignoring for the moment the dashed boxes). Although seemingly $q$ contains an $n \times n$ grid, it can be verified that the grid-like structure in the figure is actually a tree.
In addition, $q$ contains atoms of the form $R(x,y,z,w)$ as explained in the figure. More precisely, for each of the open squares occurring in the first column (e.g., the upper-left shaded square), we have the two atoms $R(x,y,z,w_1)$ and $R(x,y,z,w_2)$ represented by the two hyperedges on the left.
Moreover, for each of the internal open squares and the open squares occurring in the last column (e.g., the upper-right shaded square), we have the two atoms $R(x,y,z,w_1)$ and $R(x,y,z,w_2)$ represented by the two hyperedges on the right. Observe that $q$ is an acyclic query.
Consider now the set $\dep$ of keys:
\begin{multline*}
\epsilon_1\ =\ R(x,y,z,w), R(x,y,z,w') \ra w = w'\\
\epsilon_2\ =\ H(x,y), H(x,z) \ra y = z.
\end{multline*}
Notice that $H(\cdot,\cdot)$ stores the horizontal edges. It is not difficult to see that $\chase{q}{\dep}$ contains an $n \times n$ grid.
Roughly, as described at the bottom of Figure~\ref{fig:tree-to-grid}, by first applying $\epsilon_1$ we close the open squares of the first column, while the open squares of the second column have now the same shape as the ones of the first column, but with a dangling $H$-edge. Then, by applying $\epsilon_2$, the two $H$-edges collapse into a single edge, and we obtain open squares that have exactly the same shape as those of the first column. After finitely many chase steps, all the squares are closed, and thus $\chase{q}{\dep}$ indeed contains an $n \times n$ grid. Therefore, although the query $q$ is acyclic, $\chase{q}{\dep}$ is far from being acyclic.
Observe also that the (hyper)tree width of $\chase{q}{\dep}$ depends on $n$, while $q$ has (hyper)tree width $3$. \hfill\markfull
\end{example}

\noindent \paragraph{UCQ rewritability.} It is not hard to show that keys
are not UCQ rewritable. This is not surprising due to the transitive
nature of equality. Intuitively, the UCQ rewritability of
keys implies that a first-order (FO) query can encode the fact that
the equality relation is transitive. However, it is well-known that this is not possible due to the inability of FO queries to express recursion.

\subsection{Keys over Constrained Signatures}

Despite the peculiar nature of keys as discussed above, we can establish a positive result regarding semantic acyclicity under keys, providing that only unary and binary predicates can be used. This is done by exploiting the following generic result, which is actually the version of Theorem~\ref{theo:apc} for egd-based classes:

\begin{theorem}\label{theo:apc-egds}
Consider a class $\class{C}$ of sets of egds. If $\class{C}$ has acyclicity-preserving chase, then $\SA(\class{C})$ is \NP-complete, even if we allow only unary and binary predicates.
\end{theorem}

The proof of the above result is along the lines of the proof for Theorem~\ref{theo:apc}, and exploits the fact that the containment problem under egds is feasible in non-deterministic polynomial time (this can be shown by using Lemma~\ref{lemma:cq-equiv-tgds}).
The lower bound follows from~\cite{DaKV02}, which shows that the problem of checking whether a Boolean CQ over a single binary relation is equivalent to an acyclic one is \NP-hard.
We now show the following positive result for the class of keys over unary and binary predicates, denoted $\class{K}_2$:

\begin{proposition}\label{pro:apc-kds-unary-binary-predicates}
$\class{K}_2$ has acyclicity-preserving chase.
\end{proposition}

Notice that the above result is not in a conflict with Examples~\ref{exa:keys-violate-acyclicity} and~\ref{exa:tree-to-grid}, since both examples use predicates of arity greater than two.
It is now straightforward to see that:

\begin{theorem}\label{the:unary-binary-keys}
$\SA(\class{K}_2)$ is \NP-complete.
\end{theorem}

Interestingly, Theorem~\ref{the:unary-binary-keys} can be extended to \emph{unary functional dependencies} (over unconstrained signatures), that is, FDs of the form $R : A \ra B$, where $R$ is a relational symbol of arity $n > 0$ and the cardinality of $A$ is one. This result has been established independently by Figueira~\cite{Figu16}.
Let us recall that egds ensure the finite controllability of CQ containment. Consequently, Theorem~\ref{the:unary-binary-keys} holds even for \FSA, which takes as input a CQ $q$ and a set $\dep$ of egds, and asks for the existence of an acyclic CQ $q'$ such that $q$ and $q'$ are equivalent over all finite databases that satisfy $\dep$.

\section{Evaluation of semantically \\ acyclic queries}\label{sec:evaluation}

As it has been noted in different scenarios in
the absence of constraints, semantic
acyclicity has a positive impact on query evaluation
\cite{BaLR14,BPS15,BRV13}.
We observe here that such good behavior extends to the notion of
semantic acyclicity for CQs
under the decidable classes of constraints we
identified in the previous sections.
In particular, evaluation of semantically acyclic CQs under constraints
in such classes is a {\em
  fixed-parameter tractable} (fpt) problem, assuming the parameter to be $|q| +
|\Sigma|$. (Here, $|q|$ and $|\Sigma|$ represent the size of
reasonable encodings of $q$ and $\Sigma$, respectively). Recall that
this means that the problem can be solved in time
$O(|D|^c \cdot f(|q| + |\Sigma|))$, for $c \geq 1$ and $f$ a computable
function.

Let $\class{C}$ be a class of sets of tgds.
We define $\saeval(\class{C})$ to be the following problem:
The input consists of a set of constraints $\Sigma$ in $\class{C}$, a
semantically acyclic CQ $q$ under $\Sigma$, a database $D$ such that $D \models \Sigma$,
and a tuple $\bar t$ of
elements in $D$. We ask
whether $\bar t \in q(D)$.

\begin{proposition} \label{prop:fpt}
$\saeval(\class{C})$ can be solved in time
\[
O\left(|D| \cdot
2^{2^{O(|q| + |\Sigma|)}}\right),
\]
where $\class{C} \in \{\guarded,\nr,\sticky\}$.
\end{proposition}

\begin{proof}
Our results state that for $\class{C} \in \{\guarded,\nr,\sticky\}$,
checking if a CQ $q$ is semantically
acyclic under $\class{C}$ can be done in double-exponential
time. More importantly, in case that $q$ is in fact semantically
acyclic under $\class{C}$ our proof techniques yield an equivalent acyclic CQ
$q'$ of at most exponential size in $|q| +
|\Sigma|$.
 We then compute and evaluate such a query $q'$ on $D$, and return $q(D) = q'(D)$. Clearly, this can be done in
time
\[
O\left(2^{2^{O(|q| + |\Sigma|)}}\right)\ +\ O\left(|D| \cdot
2^{O(|q| + |\Sigma|)}\right).
\]
The running time of this algorithm is dominated by
\[
O\left(|D| \cdot 2^{2^{O(|q| + |\Sigma|)}}\right)
\]
and the claim follows.
\end{proof}

This is an improvement over general CQ
evaluation for which no fpt
algorithm is believed to exist \cite{PY99}. It is worth remarking,
nonetheless,
that $\saeval(\class{C})$ corresponds to
a {\em promise version} of the evaluation problem, where
the property that defines the class is
\exptime-hard for all the $\class{C}$'s  studied in Proposition \ref{prop:fpt}.

The above algorithm computes
an equivalent acyclic CQ $q'$ for a semantically
acyclic CQ $q$ under a set of constraints in $\class{C}$. This might take
double-exponential time. Surprisingly,
computing such $q'$ is not always needed at the moment of evaluating semantically
acyclic CQs under constraints.
In particular, this holds for
the sets of guarded tgds. In fact,
in such case evaluating a semantically acylic CQ $q$ under
$\Sigma$ over a database $D$ that satisfies $\Sigma$
amounts to checking a polynomial time property over $q$ and $D$. It
follows, in addition, that the evaluation problem for semantically acyclic CQs under
guarded tgds is tractable:

\begin{theorem} \label{theo:eval-guarded}
$\saeval(\guarded)$  is in polynomial time.
\end{theorem}

The idea behind the proof of the above theorem is as follows.
When $q$ is a semantically acyclic CQ in the absence of
constraints, evaluating $q$ on $D$ amounts to
checking the existence of a winning strategy for the duplicator in a
particular version of the pebble game, known as the {\em existential
  1-cover game}, on $q$ and $D$ \cite{ChenD05}.
We denote this by $q \equiv_{\exists 1 c} D$.
The existence of
such winning strategy can be checked in polynomial time \cite{ChenD05}. Now, when
$q$ is semantically acyclic under an arbitrary set $\Sigma$ of tgds or egds, we show that
evaluating $q$ on $D$ amounts to checking whether $\chase{q}{\Sigma}
\equiv_{\exists 1 c} D$. When $\Sigma$ is a set of guarded tgds, we
prove in addition that $\chase{q}{\Sigma}
\equiv_{\exists 1 c} D$ iff $q \equiv_{\exists 1 c} D$.
Thus, $\saeval(\guarded)$ is tractable since
checking $q \equiv_{\exists1c} D$
is tractable.

The fact that the evaluation of $q$ on $D$ boils down to
checking whether $\chase{q}{\Sigma}
\equiv_{\exists 1 c} D$, when $q$ is semantically acyclic under $\Sigma$,
also yields tractability for $\saeval(\class{C})$,
where $\class{C}$ is any class of sets of egds for which the chase can be computed in
polynomial time. This includes the central class of FDs. Notice, however, that
we do not know whether $\SA$ under FDs is decidable. 
\section{Further Advancements}\label{sec:further-advancements}

In this section we informally discuss the fact that our previous results on semantic acyclicity under tgds and CQs can be extended to UCQs. Moreover, we show that our techniques establish the existence of maximally contained acyclic queries. 

\subsection{Unions of Conjunctive Queries}

It is reasonable to consider a more {\em liberal} version of semantic acyclicity under tgds based on UCQs. In such case we are given a UCQ $Q$ and a finite set $\dep$ of tgds, and the question is whether there is a union $Q'$ of acyclic CQs that is equivalent to $Q$ under $\Sigma$.
It can be shown that all the complexity results on semantic acyclicity under tgds presented above continue to hold even when the input query is a UCQ. This is done by extending the small query properties established for CQs (Propositions~\ref{prop:apc-small-query-property} and~\ref{prop:ucq-rewritability-small-query-property}) to UCQs.

Consider a finite set $\dep$ of tgds (that falls in one of the tgd-based classes considered above), and a UCQ $Q$. If $Q$ is semantically acyclic under $\dep$, then one of the following holds: (i) for each disjunct $q$ of $Q$, there exists an acyclic CQ $q'$ of bounded size (the exact size of $q'$ depends on the class of $\dep$) such that $q \equiv_{\dep} q'$, or (ii) $q$ is redundant in $Q$, i.e., there exists a disjunct $q'$ of $Q$ such that $q \subseteq_{\dep} q'$.
Having this property in place, we can then design a non-deterministic algorithm for semantic acyclicity, which provides the desired upper bounds. Roughly, for each disjunct $q$ of $Q$, this algorithm guesses whether there exists an acyclic CQ $q'$ of bounded size such that $q \equiv_{\dep} q'$, or $q$ is redundant in $Q$.
The desired lower bounds are inherited from semantic acyclicity in the case of CQs.

\subsection{Query Approximations}

Let $\class{C}$ be any of the decidable classes of finite sets of tgds we
study in this paper (i.e., $\guarded$, $\nr$, or $\sticky$). Then, for any CQ $q$ without constants\footnote{Approximations for CQs with constants are
  not well-understood, even in the absence of constraints \cite{BaLR14}.}
and set $\Sigma$ of constraints in $\class{C}$, our techniques yield a {\em
  maximally contained} acyclic CQ $q'$ under $\Sigma$. This means
that $q' \subseteq_\Sigma q$ and there is no acyclic CQ $q''$ such that
$q'' \subseteq_\Sigma q$ and $q' \subsetneq_\Sigma q''$.
Following the recent database literature,
such $q'$ corresponds to an {\em acyclic CQ
  approximation of $q$ under $\Sigma$}
\cite{BaLR14,BPS15,BRV13}. Notice that when $q$ is semantically
acyclic under $\Sigma$, this acyclic approximation $q'$ is in fact
equivalent to $q$ under $\Sigma$. Computing and evaluating an acyclic
CQ approximation for $q$ might help finding ``quick'' (i.e., fixed-parameter
tractable) answers to it when exact evaluation is infeasible.

The way in which we obtain approximations is by slightly
reformulating the small query properties established in the paper
(Propositions~\ref{prop:apc-small-query-property} and~\ref{prop:ucq-rewritability-small-query-property}). Instead of dealing with semantically acyclic CQs only, we are now given an arbitrary CQ
$q$. In all cases the reformulation expresses the following:
For every acyclic CQ $q'$ such that $q' \subseteq_\Sigma q$,
there is an acyclic CQ $q''$ of the appropriate size $f(q,\Sigma)$
such that $q' \subseteq_\Sigma q'' \subseteq_\Sigma q$. It is easy to
prove that for each CQ $q$ there exists at least one acyclic CQ $q'$
such that $q' \subseteq_\Sigma q$; take a single variable $x$ and
add a tuple $R(x,\dots,x)$ for each symbol $R$. It follows
then that an acyclic CQ approximation of $q$ under $\Sigma$ can always
be found among the set $\A(q)$ of acyclic CQs $q'$ of size at most $f(q,\Sigma)$ such that $q' \subseteq_\Sigma q$. In fact, the acyclic CQ approximations of $q$ under $\Sigma$ are the maximal
elements of $\A(q)$ under $\subseteq_\Sigma$.

\section{Conclusions}\label{sec:conclusions}

We have concentrated on the problem of semantic acyclicity for CQs in the presence of database constraints; in fact, tgds or egds. Surprisingly, we have shown that there are cases such as the class of full tgds, where containment is decidable, while semantic acyclicity is undecidable.
We have then focussed on the main classes of tgds for which CQ containment is decidable, and do not subsume full tgds, i.e., guarded, non-recusrive and sticky tgds. For these classes we have shown that semantic acyclicity is decidable, and obtained several complexity results. We have also shown that semantic acyclicity is \NP-complete if we focus on keys over unary and binary predicates.
Finally, we have considered the problem of evaluating a semantically acyclic CQ over a database that satisfies certain constraints, and shown that for guarded tgds and FDs is tractable.
Here are some interesting open problems that we are planning to investigate: (i) The complexity of semantic acyclicity under sticky sets of tgds is still unknown;
(ii) We do not know whether semantic acyclicity under keys over unconstrained signatures is decidable; and (iii) We do not know the complexity of evaluating semantically acyclic queries under $\class{NR}$, $\class{S}$ and egds.

\smallskip
\noindent
{\small
{\bf Acknowledgements:} Barcel\'o would like to thank
D. Figueira, M. Romero, S. Rudolph, and N. Schweikardt for insightful
discussions about the nature of
semantic acyclicity under constraints.
}

\bibliographystyle{abbrv}

\onecolumn

\newpage

\section*{Appendix}

\subsection*{Proof of Proposition~\ref{prop:gral-reduction}}

It is not difficult to show the following result, which reveals the usefulness of connectedness:

\begin{lemma}\label{lem:distribution-over-components}
Let $\dep$ be a finite set of body-connected tgds, $q$ a Boolean CQ, and $q'$ a Boolean and connected CQ. If $q \subseteq_{\dep} q'$, then there exists a maximally connected subquery $q''$ of $q$ such that $q'' \subseteq_{\dep} q'$.
\end{lemma}

Having the above result in place, we can now establish Proposition~\ref{prop:gral-reduction}. For brevity, let $q''$ be the CQ $q \wedge q'$.

\medskip

$(\Rightarrow)$ It is clear that $q \subseteq_{\dep} q''$. Moreover, $q'' \subseteq_{\dep} q$ holds trivially. Therefore, $q'' \equiv_{\dep} q$, and the claim follows since, by hypothesis, $q$ is acyclic.

\medskip

$(\Leftarrow)$ Since $\dep$ belongs to a class that ensures the finite controllability of containment, it suffices to show the following: If there exists an acyclic Boolean CQ $q_A$ such that $q'' \equiv_{\dep} q_A$, then $q \subseteq_{\dep} q'$.
Let $q_{A}^{1},\ldots,q_{A}^{k}$, where $k \geqslant 1$, be the maximally connected subqueries of $q_A$. Clearly, $q$ and $q'$ are the two maximally connected subqueries of $q''$. Therefore, by Lemma~\ref{lem:distribution-over-components}, for each $i \in \{1,\ldots,k\}$, $q \subseteq_{\dep} q_{A}^{i}$ or $q' \subseteq_{\dep} q_{A}^{i}$.
We define the following two sets of indices:
\[
S_q\ =\ \{i \in \{1,\ldots,k\} \mid q \subseteq_{\dep} q_{A}^{i}\} \qquad\textrm{and}\qquad S_{q'}\ =\ \{i \in \{1,\ldots,k\} \mid q' \subseteq_{\dep} q_{A}^{i} \textrm{~~and~~} q \not\subseteq_{\dep} q_{A}^{i}\};
\]
clearly, $S_q$ and $S_{q'}$ form a partition of $\{1,\ldots,k\}$. We proceed to show that $q \subseteq_{\dep} q'$ by considering the following three cases:

\medskip

\textbf{Case 1.} Assume first that $S_q = \emptyset$. This implies that, for each $i \in \{1,\ldots,k\}$, $q' \subseteq_{\dep} q_{A}^{i}$; thus, $q' \subseteq_{\dep} q_{A}$. By hypothesis, $q_A \subseteq_{\dep} q''$, which immediately implies that $q_A \subseteq_{\dep} q'$. Therefore, $q' \equiv_{\dep} q_A$, which contradicts our hypothesis that $q'$ is not semantically acyclic under $\dep$. Hence, the case where $S_q = \emptyset$ is not possible, and is excluded from our analysis.

\medskip

\textbf{Case 2.} Assume now that $S_{q'} = \emptyset$. This implies that, for each $i \in \{1,\ldots,k\}$, $q \subseteq_{\dep} q_{A}^{i}$; thus, $q \subseteq_{\dep} q_A$. By hypothesis, $q_A \subseteq_{\dep} q''$, which immediately implies that $q_A \subseteq_{\dep} q'$. Therefore, $q \subseteq_{\dep} q'$, as needed.

\medskip

\textbf{Case 3.} Finally, assume that $S_q \neq \emptyset$ and $S_{q'} \neq \emptyset$. Fix an arbitrary index $i \in S_{q'}$. Since $q' \subseteq_{\dep} q_{A}^{i}$ and $q'$ is not semantically acyclic under $\dep$, we conclude that $q_{A}^{i} \not\subseteq_{\dep} q'$; notice that $q_{A}^{i}$ is necessarily acyclic. Since $q_A \subseteq_{\dep} q'$, $\dep$ is body-connected and $q'$ is connected, Lemma~\ref{lem:distribution-over-components} implies that there exists $j \in \{1,\ldots,k\} \setminus \{i\}$ such that $q_{A}^{j} \subseteq_{\dep} q'$. Observe that $j \not\in S_{q'}$; otherwise, $q' \subseteq_{\dep} q_{A}^{j}$, and thus $q' \equiv_{\dep} q_{A}^{j}$, which contradicts the fact that $q'$ is not semantically acyclic under $\dep$. Since $j \in S_q$, $q \subseteq_{\dep} q_{A}^{j}$, and therefore $q \subseteq_{\dep} q'$, as needed. This completes our proof.

\subsection*{Proof of Theorem \ref{theo:full-undec}}

We reduce from the {\em Post correspondence problem} (PCP)
over alphabet $\Sigma = \{a,b\}$. Recall that the input to this problem are two
equally long lists $w_1,\dots,w_n$ and $w'_1,\dots,w'_n$
of words over $\Sigma$, and we ask whether there is a {\em solution}, i.e., a nonempty sequence
$i_1 \dots i_m$ of indices in $\{1,\dots,n\}$, such that $w_{i_1} \dots w_{i_m} =
w'_{i_1} \dots w'_{i_m}$. We assume without loss of generality that all
these words are of even length. Otherwise we simply replace in each word
each appearance of $a$ (resp., $b$) with $aa$ (resp., $bb$).

Let $w_1,\dots,w_n$ and $w'_1,\dots,w'_n$ be an instance of PCP as described above.
We construct a Boolean CQ $q$ and a
set of full tgds $\Sigma$ over the schema: $$\{P_a,P_b,P_\#,P_*,{\rm sync},{\rm
  start},{\rm end}\},$$
where $P_a$, $P_b$, $P_\#$, $P_*$, and ${\rm sync}$ are binary relation
symbols, and the other ones are unary relation symbols,
such that the PCP instance given by $w_1,\dots,w_n$ and $w'_1,\dots,w'_n$ has a solution
if and only if there exists an acylic CQ $q'$
such that $q \equiv_\Sigma q'$.

We start with a temporary version of $q$.
This query will have to be modified later in order to
make the proof work, but it is convenient to work with this version
for the time being in order to simplify the presentation.
The restriction of our Boolean CQ $q$ to those relation symbols that
are not ${\rm sync}$ is graphically defined as follows:
\begin{center}
\begin{picture}(0,0)%
\includegraphics{cq1.pstex}%
\end{picture}%
\setlength{\unitlength}{2368sp}%
\begingroup\makeatletter\ifx\SetFigFont\undefined%
\gdef\SetFigFont#1#2#3#4#5{%
  \reset@font\fontsize{#1}{#2pt}%
  \fontfamily{#3}\fontseries{#4}\fontshape{#5}%
  \selectfont}%
\fi\endgroup%
\begin{picture}(4523,3027)(3961,-4780)
\put(8326,-3586){\makebox(0,0)[lb]{\smash{{\SetFigFont{7}{8.4}{\familydefault}{\mddefault}{\updefault}{\color[rgb]{0,0,0}${\rm end}$}%
}}}}
\put(4726,-2761){\makebox(0,0)[lb]{\smash{{\SetFigFont{7}{8.4}{\familydefault}{\mddefault}{\updefault}{\color[rgb]{0,0,0}$P_\#$}%
}}}}
\put(4801,-3286){\makebox(0,0)[lb]{\smash{{\SetFigFont{7}{8.4}{\familydefault}{\mddefault}{\updefault}{\color[rgb]{0,0,0}$P_\#$}%
}}}}
\put(4726,-4036){\makebox(0,0)[lb]{\smash{{\SetFigFont{7}{8.4}{\familydefault}{\mddefault}{\updefault}{\color[rgb]{0,0,0}$P_\#$}%
}}}}
\put(5476,-3286){\makebox(0,0)[lb]{\smash{{\SetFigFont{7}{8.4}{\familydefault}{\mddefault}{\updefault}{\color[rgb]{0,0,0}$P_a$}%
}}}}
\put(6976,-4486){\makebox(0,0)[lb]{\smash{{\SetFigFont{7}{8.4}{\familydefault}{\mddefault}{\updefault}{\color[rgb]{0,0,0}$P_a$}%
}}}}
\put(6376,-3961){\makebox(0,0)[lb]{\smash{{\SetFigFont{7}{8.4}{\familydefault}{\mddefault}{\updefault}{\color[rgb]{0,0,0}$P_b$}%
}}}}
\put(6001,-3286){\makebox(0,0)[lb]{\smash{{\SetFigFont{7}{8.4}{\familydefault}{\mddefault}{\updefault}{\color[rgb]{0,0,0}$P_b$}%
}}}}
\put(6676,-1936){\makebox(0,0)[lb]{\smash{{\SetFigFont{7}{8.4}{\familydefault}{\mddefault}{\updefault}{\color[rgb]{0,0,0}$P_a$}%
}}}}
\put(6376,-2836){\makebox(0,0)[lb]{\smash{{\SetFigFont{7}{8.4}{\familydefault}{\mddefault}{\updefault}{\color[rgb]{0,0,0}$P_b$}%
}}}}
\put(7576,-2986){\makebox(0,0)[lb]{\smash{{\SetFigFont{7}{8.4}{\familydefault}{\mddefault}{\updefault}{\color[rgb]{0,0,0}$P_*$}%
}}}}
\put(7426,-3286){\makebox(0,0)[lb]{\smash{{\SetFigFont{7}{8.4}{\familydefault}{\mddefault}{\updefault}{\color[rgb]{0,0,0}$P_*$}%
}}}}
\put(7501,-3586){\makebox(0,0)[lb]{\smash{{\SetFigFont{7}{8.4}{\familydefault}{\mddefault}{\updefault}{\color[rgb]{0,0,0}$P_*$}%
}}}}
\put(3976,-3511){\makebox(0,0)[lb]{\smash{{\SetFigFont{7}{8.4}{\familydefault}{\mddefault}{\updefault}{\color[rgb]{0,0,0}${\rm start}$}%
}}}}
\put(4351,-3211){\makebox(0,0)[lb]{\smash{{\SetFigFont{7}{8.4}{\familydefault}{\mddefault}{\updefault}{\color[rgb]{0,0,0}$x$}%
}}}}
\put(5476,-2086){\makebox(0,0)[lb]{\smash{{\SetFigFont{7}{8.4}{\familydefault}{\mddefault}{\updefault}{\color[rgb]{0,0,0}$y$}%
}}}}
\put(7276,-3511){\makebox(0,0)[lb]{\smash{{\SetFigFont{7}{8.4}{\familydefault}{\mddefault}{\updefault}{\color[rgb]{0,0,0}$u$}%
}}}}
\put(5401,-4711){\makebox(0,0)[lb]{\smash{{\SetFigFont{7}{8.4}{\familydefault}{\mddefault}{\updefault}{\color[rgb]{0,0,0}$z$}%
}}}}
\put(8326,-3211){\makebox(0,0)[lb]{\smash{{\SetFigFont{7}{8.4}{\familydefault}{\mddefault}{\updefault}{\color[rgb]{0,0,0}$v$}%
}}}}
\end{picture}%

\end{center}
Here, $x,y,z,u,v$ denote the names of the respective variables. The
interpretation of ${\rm sync}$ in $q$ consists of all pairs in $\{y,u,z\}$.

Our set $\Sigma$ of full tgds consists of the following:
\begin{enumerate}
\item An initialization rule:
$${\rm start}(x),P_\#(x,y) \, \rightarrow \, {\rm sync}(y,y).$$
\item A synchronization rule:
$${\rm sync}(x,y),P_{w_i}(x,z),P_{w'_i}(y,u) \, \rightarrow \, {\rm
  sync}(z,u),$$
for each $1 \leq i \leq n$. Here, $P_{w}(x,y)$, for a nonepty word $w = a_1 \dots a_t \in
\Sigma^*$, is a shortening for
$P_{a_1}(x,x_1),\dots,P_{a_t}(x_{t-1},y)$, where the $x_i$'s are fresh
variables.
\item For each $1 \leq i \leq n$, a pair of finalization rules defined as follows.
First:
\begin{multline*}
{\rm start}(x),P_a(y,z),P_a(z,u),{\rm sync}(y',y''),P_{w_i}(y',y),P_{w'_i}(y'',y),
P_*(u,v),{\rm end}(v) \, \rightarrow
 \\
P_\#(x,y),P_\#(x,z),P_\#(x,u),P_*(y,v),P_*(z,v),P_b(z,y),P_b(u,z),P_a(u,y),P_b(y,u).
\end{multline*}
Second:
\begin{multline*}
{\rm start}(x),P_a(y,z),P_a(z,u),{\rm sync}(y',y''),P_{w_i}(y',y),P_{w'_i}(y'',y),
P_*(u,v),{\rm end}(v) \, \rightarrow
 \\
{\rm sync}(y,y),{\rm sync}(z,z),{\rm sync}(y,z),{\rm sync}(z,y),{\rm sync}(y,u),{\rm sync}(u,y),{\rm sync}(z,u),{\rm sync}(u,z).
\end{multline*}
Notice that these two tgds can be expressed into one since they have the same
body. We split in two for the sake of presentation.
\end{enumerate}

We first show that if the PCP instance has a solution given by the
nonempty sequence $i_1 \dots i_m$, with $1 \leq i_1,\dots,i_m \leq n$,
then there exists an acylic CQ $q'$
such that $q \equiv_\Sigma q'$. Let us assume that
$w_{i_1} \dots w_{i_m} = a_1 \dots a_t$, where each $a_i$ is a symbol
in $\Sigma$. We prove next that $q \equiv_\Sigma q'$, where
$q'$ is the Boolean acyclic CQ depicted below:
\begin{center}
\begin{picture}(0,0)%
\includegraphics{acq.pstex}%
\end{picture}%
\setlength{\unitlength}{2368sp}%
\begingroup\makeatletter\ifx\SetFigFont\undefined%
\gdef\SetFigFont#1#2#3#4#5{%
  \reset@font\fontsize{#1}{#2pt}%
  \fontfamily{#3}\fontseries{#4}\fontshape{#5}%
  \selectfont}%
\fi\endgroup%
\begin{picture}(7373,570)(2911,-1846)
\put(9601,-1486){\makebox(0,0)[lb]{\smash{{\SetFigFont{6}{7.2}{\familydefault}{\mddefault}{\updefault}{\color[rgb]{0,0,0}$P_{*}$}%
}}}}
\put(2926,-1786){\makebox(0,0)[lb]{\smash{{\SetFigFont{6}{7.2}{\familydefault}{\mddefault}{\updefault}{\color[rgb]{0,0,0}${\rm start}$}%
}}}}
\put(3451,-1486){\makebox(0,0)[lb]{\smash{{\SetFigFont{6}{7.2}{\familydefault}{\mddefault}{\updefault}{\color[rgb]{0,0,0}$P_{\#}$}%
}}}}
\put(4351,-1486){\makebox(0,0)[lb]{\smash{{\SetFigFont{6}{7.2}{\familydefault}{\mddefault}{\updefault}{\color[rgb]{0,0,0}$P_{a_1}$}%
}}}}
\put(3151,-1411){\makebox(0,0)[lb]{\smash{{\SetFigFont{6}{7.2}{\familydefault}{\mddefault}{\updefault}{\color[rgb]{0,0,0}$x'$}%
}}}}
\put(10051,-1786){\makebox(0,0)[lb]{\smash{{\SetFigFont{6}{7.2}{\familydefault}{\mddefault}{\updefault}{\color[rgb]{0,0,0}${\rm end}$}%
}}}}
\put(7501,-1411){\makebox(0,0)[lb]{\smash{{\SetFigFont{6}{7.2}{\familydefault}{\mddefault}{\updefault}{\color[rgb]{0,0,0}$y'$}%
}}}}
\put(8401,-1411){\makebox(0,0)[lb]{\smash{{\SetFigFont{6}{7.2}{\familydefault}{\mddefault}{\updefault}{\color[rgb]{0,0,0}$z'$}%
}}}}
\put(9301,-1411){\makebox(0,0)[lb]{\smash{{\SetFigFont{6}{7.2}{\familydefault}{\mddefault}{\updefault}{\color[rgb]{0,0,0}$u'$}%
}}}}
\put(10201,-1411){\makebox(0,0)[lb]{\smash{{\SetFigFont{6}{7.2}{\familydefault}{\mddefault}{\updefault}{\color[rgb]{0,0,0}$v'$}%
}}}}
\put(6826,-1486){\makebox(0,0)[lb]{\smash{{\SetFigFont{6}{7.2}{\familydefault}{\mddefault}{\updefault}{\color[rgb]{0,0,0}$P_{a_t}$}%
}}}}
\put(5476,-1561){\makebox(0,0)[lb]{\smash{{\SetFigFont{11}{13.2}{\familydefault}{\mddefault}{\updefault}{\color[rgb]{0,0,0}$\dots$}%
}}}}
\put(7801,-1486){\makebox(0,0)[lb]{\smash{{\SetFigFont{6}{7.2}{\familydefault}{\mddefault}{\updefault}{\color[rgb]{0,0,0}$P_{a}$}%
}}}}
\put(8701,-1486){\makebox(0,0)[lb]{\smash{{\SetFigFont{6}{7.2}{\familydefault}{\mddefault}{\updefault}{\color[rgb]{0,0,0}$P_{a}$}%
}}}}
\end{picture}%

\end{center}
Here, again, $x',y',z',u',v'$ denote the names of the respective
variables of $q'$.  It is clear
that all these elements are different.

 In virtue of Lemma \ref{lemma:cq-equiv-tgds}, we only need to show
 that $\chase{q}{\Sigma} \equiv \chase{q'}{\Sigma}$. It is well-known
 that the latter is equivalent to showing that $\chase{q}{\Sigma}$ and
 $\chase{q'}{\Sigma}$ are homomorphically equivalent \cite{ChMe77}.
 Let us start by analyzing what $\chase{q}{\Sigma}$ and
 $\chase{q'}{\Sigma}$ are: \begin{enumerate} \item Notice that
 $\chase{q'}{\Sigma}$ extends $q'$ with the atom ${\rm
 sync}(x'',x'')$, where $x''$ is the second element of $q'$ (i.e., the
 successor of $x'$), plus all atoms of the form ${\rm sync}(y^*,z^*)$
 produced by the recursive applications of the second rule starting
 from ${\rm sync}(x'',x'')$.  Further, since $w_{i_1} \dots w_{i_m} =
 w'_{i_1} \dots w'_{i_m} = a_1 \dots a_t$, it must be the case that
 the atom ${\rm sync}(y',y')$ is generated in this process. Thus,
 there are elements $y_1$ and $y_2$ such that ${\rm sync}(y_1,y_2)$,
 $P_{w_{i_m}}(y_1,y')$ and $P_{w'_{i_m}}(y_2,y')$ hold. From the third
 rule
we conclude that
 that $\chase{q'}{\Sigma}$ contains the atoms in the following
sets. First: $$S_1 \ = \ \{P_\#(x',y'),P_\#(x',z'),P_\#(x',u'),P_*(y',v'),P_*(z',v'),
P_b(z',y'),P_b(u',z'),P_a(u',y'),P_b(y',u')\}.$$
Second:
$$S_2 \ = \ \{{\rm sync}(y',y'),{\rm sync}(z',z'),{\rm sync}(y',z'),{\rm sync}(z',y'),{\rm
sync}(y',u'),{\rm sync}(u',y'),{\rm sync}(z',u'),{\rm sync}(u',z')\}.$$

\item It can be checked that $q$ is {\em closed} under $\Sigma$, i.e.,
$q = \chase{q}{\Sigma}$.

\end{enumerate}

We show first that $\chase{q'}{\Sigma}$ can be homomorphically
mapped to $q = \chase{q}{\Sigma}$. But this is easy; we simply map
the variable $x'$ to $x$, the variable $v'$ to $v$,
and every consecutive node in between $x'$ and $v'$ in
$q'$ to the corresponding element in between $x$ and $v$ in $q$.

Let us show now that $q = \chase{q}{\Sigma}$ can be homomorphically
mapped to $\chase{q'}{\Sigma}$. We use the mapping that sends
$x,y,z,u,v$ to $x',y',z',u',v'$, respectively.
It is not hard to check that this mapping is a homomorphism using the
fact that $S_1,S_2 \subseteq \chase{q'}{\Sigma}$.

Now we prove that if there exists an acyclic CQ $q'$ such that $q
\equiv_\Sigma q'$, then there are indices
$1 \leq i_1,\dots,i_m \leq n$ such that $w_{i_1} \dots w_{i_m} =
w'_{i_1} \dots w'_{i_m}$. We start with a simpler case. We assume that
the restriction of $q'$ to $P_a$, $P_b$, $P_\#$, $P_\#$
and ${\rm sync}$ looks like this:
\begin{center}
\begin{picture}(0,0)%
\includegraphics{sc.pstex}%
\end{picture}%
\setlength{\unitlength}{2368sp}%
\begingroup\makeatletter\ifx\SetFigFont\undefined%
\gdef\SetFigFont#1#2#3#4#5{%
  \reset@font\fontsize{#1}{#2pt}%
  \fontfamily{#3}\fontseries{#4}\fontshape{#5}%
  \selectfont}%
\fi\endgroup%
\begin{picture}(4666,294)(2018,-1648)
\put(4501,-1561){\makebox(0,0)[lb]{\smash{{\SetFigFont{11}{13.2}{\familydefault}{\mddefault}{\updefault}{\color[rgb]{0,0,0}$\dots$}%
}}}}
\end{picture}%

\end{center}
That is, the underlying graph of this query corresponds to a directed
path.

Since $q
\equiv_\Sigma q'$, we can conclude from Lemma
\ref{lemma:cq-equiv-tgds} that $\chase{q}{\Sigma}
\equiv \chase{q'}{\Sigma}$, i.e.,  $\chase{q}{\Sigma}$ and
$\chase{q'}{\Sigma}$ are homomorphically equivalent. But then
$\chase{q'}{\Sigma}$ must contain at least one variable labeled ${\rm
  start}$ and one variable labeled ${\rm end}$. The first variable cannot have incoming edges (otherwise,
$\chase{q'}{\Sigma}$ would not homomorphically map to
$\chase{q}{\Sigma}$), while the second one cannot have outcoming edges
(for the same reason). This implies that it is the first variable $x'$ of
$q'$ that is labeled ${\rm start}$ and it is the last one $v'$ that is
labeled ${\rm end}$. Furthermore, all edges
reaching $v'$ in $q'$ must be labeled
$P_*$ (otherwise $q'$ does not homomorphically map to $q$).
Thus, this is the label of the last edge of $q'$ that goes
from variable $u'$ to $v'$. Analogously, the edge that leaves $x'$ in
$q'$ is labeled $P_\#$. Furthermore, any other edge in $q'$ must be
labeled $P_a$, $P_b$, or ${\rm sync}$.

Notice now that $v'$ must have an incoming edge labeled $P_*$ in
$\chase{q'}{\Sigma}$ from
some node $u''$ that has an outgoing edge with label $P_a$ (since $q$ homomorphically maps
to $\chase{q'}{\Sigma}$).
By the definition of $\Sigma$, this could
only have happened if there are elements $y_1$ and $y_2$ such that the
following atoms hold in $\chase{q'}{\Sigma}$, for some $1 \leq i \leq n$:
$$\{P_a(y',z'),P_a(z',u'),{\rm sync}(y_1,y_2),P_{w_i}(y_1,y'),P_{w'_i}(y_2,y')\},$$
where $y',z',u'$ are the immediate predecessors of $v'$ in the order
that is naturally induced by $q'$. In particular, the
unique path from $y_1$ (resp., $y_2$) to $y'$ in $q'$ is labeled $w_i$
(resp., $w'_i$). This means that the atom ${\rm sync}(y_1,y_2)$ was
not one of the edges of $q'$, but must have been produced during the
chase by another atom of this form, and so on.
This process can only finish in the second element $x^*$ of
$q'$ (notice that ${\rm sync}(x^*,x^*)$ belongs to $\chase{q'}{\Sigma}$
due to the first rule of $\Sigma$).
We conclude then that our PCP instance has a solution.

What now, if $q'$ contains parallel edges going from one node to
another (but in the same direction than before)? Notice that $q$
contains no parallel edges save for those in between the elements in
$\{y,z,u\}$. These parallel edges
are labeled with both $P_a$ (or $P_b$) and ${\rm
sync}$. Thus, parallel edges in $q'$ can only be of this form (since
$q'$ homomorphically maps to $q$). This implies that $q'$ can now
contain edges labeled ${\rm sync}$ (this was not the case before). On
the other hand, there can be at most one edge labeled in $\{P_a,P_b\}$
from one node to another in $q'$. This is crucial for our reduction to
hold.

We use the same idea than before. We know that
${\rm sync}(y_1,y_2)$, $P_{w_i}(y_1,y')$, and $P_{w'_i}(y_2,y')$ hold
in $\chase{q'}{\Sigma}$.
Thus, if we now
restrict $q'$ to relation symbols $P_a$ and $P_b$,
there is a unique path from $y^*$ (resp., $z^*$) to $u'$ in $q'$, and
such path is labeled $w_i$ (resp., $w'_i$). Now, the question is
whether ${\rm sync}(y_1,y_2)$ could have been part of $q'$ or was
produced by the chase. If the former was the case, we would have that
$y_1$ and $y_2$ are at distance one, and, therefore, that $|w_i| =
|w'_i| + 1$ (or viceversa). But this is not possible since we are assumming both
$w_i$ and $w'_i$ to be of even length. Thus, ${\rm sync}(y_1,y_2)$
needs to have been produced by the chase. Iterating this process takes
us again all the way back to the atom ${\rm sync}(x^*,x^*)$. We thus
conclude again that the PCP instance given by $w_1,\dots,w_n$ and
$w'_1,\dots,w'_n$ has a solution.

Let us suppose now that $q'$ contains parallel edges and some of these
edges also go in the opposite direction than the ones we have
now. This is problematic for our reduction since now words in this
path can be read in both directions. This is why we mentioned in the
beginning of the proof that our version of $q$ was only temporary, and
that we would have to change it later in order to make the proof
work. The restriction of $q$ to those relation symbols that
are neither $P_a$, $P_b$, nor ${\rm sync}$ will now look like this:
\begin{center}
\begin{picture}(0,0)%
\epsfig{file=cq2.pstex}%
\end{picture}%
\setlength{\unitlength}{2368sp}%
\begingroup\makeatletter\ifx\SetFigFont\undefined%
\gdef\SetFigFont#1#2#3#4#5{%
  \reset@font\fontsize{#1}{#2pt}%
  \fontfamily{#3}\fontseries{#4}\fontshape{#5}%
  \selectfont}%
\fi\endgroup%
\begin{picture}(4851,2839)(3976,-4769)
\put(5101,-3061){\makebox(0,0)[lb]{\smash{{\SetFigFont{7}{8.4}{\familydefault}{\mddefault}{\updefault}{\color[rgb]{0,0,0}$P_\#$}%
}}}}
\put(4726,-2761){\makebox(0,0)[lb]{\smash{{\SetFigFont{7}{8.4}{\familydefault}{\mddefault}{\updefault}{\color[rgb]{0,0,0}$P_\#$}%
}}}}
\put(4801,-3286){\makebox(0,0)[lb]{\smash{{\SetFigFont{7}{8.4}{\familydefault}{\mddefault}{\updefault}{\color[rgb]{0,0,0}$P_\#$}%
}}}}
\put(4726,-4036){\makebox(0,0)[lb]{\smash{{\SetFigFont{7}{8.4}{\familydefault}{\mddefault}{\updefault}{\color[rgb]{0,0,0}$P_\#$}%
}}}}
\put(7576,-2986){\makebox(0,0)[lb]{\smash{{\SetFigFont{7}{8.4}{\familydefault}{\mddefault}{\updefault}{\color[rgb]{0,0,0}$P_*$}%
}}}}
\put(7426,-3286){\makebox(0,0)[lb]{\smash{{\SetFigFont{7}{8.4}{\familydefault}{\mddefault}{\updefault}{\color[rgb]{0,0,0}$P_*$}%
}}}}
\put(7501,-3586){\makebox(0,0)[lb]{\smash{{\SetFigFont{7}{8.4}{\familydefault}{\mddefault}{\updefault}{\color[rgb]{0,0,0}$P_*$}%
}}}}
\put(3976,-3511){\makebox(0,0)[lb]{\smash{{\SetFigFont{7}{8.4}{\familydefault}{\mddefault}{\updefault}{\color[rgb]{0,0,0}${\rm start}$}%
}}}}
\put(8326,-3211){\makebox(0,0)[lb]{\smash{{\SetFigFont{7}{8.4}{\familydefault}{\mddefault}{\updefault}{\color[rgb]{0,0,0}${\rm end}$}%
}}}}
\put(8326,-3661){\makebox(0,0)[lb]{\smash{{\SetFigFont{7}{8.4}{\familydefault}{\mddefault}{\updefault}{\color[rgb]{0,0,0}${\rm success}$}%
}}}}
\put(4351,-3211){\makebox(0,0)[lb]{\smash{{\SetFigFont{7}{8.4}{\familydefault}{\mddefault}{\updefault}{\color[rgb]{0,0,0}$x$}%
}}}}
\put(5476,-2086){\makebox(0,0)[lb]{\smash{{\SetFigFont{7}{8.4}{\familydefault}{\mddefault}{\updefault}{\color[rgb]{0,0,0}$y$}%
}}}}
\put(7276,-3511){\makebox(0,0)[lb]{\smash{{\SetFigFont{7}{8.4}{\familydefault}{\mddefault}{\updefault}{\color[rgb]{0,0,0}$u$}%
}}}}
\put(8551,-3436){\makebox(0,0)[lb]{\smash{{\SetFigFont{7}{8.4}{\familydefault}{\mddefault}{\updefault}{\color[rgb]{0,0,0}$v$}%
}}}}
\put(5401,-4711){\makebox(0,0)[lb]{\smash{{\SetFigFont{7}{8.4}{\familydefault}{\mddefault}{\updefault}{\color[rgb]{0,0,0}$z$}%
}}}}
\put(7726,-3886){\makebox(0,0)[lb]{\smash{{\SetFigFont{7}{8.4}{\familydefault}{\mddefault}{\updefault}{\color[rgb]{0,0,0}$P_*$}%
}}}}
\put(7801,-2836){\makebox(0,0)[lb]{\smash{{\SetFigFont{7}{8.4}{\familydefault}{\mddefault}{\updefault}{\color[rgb]{0,0,0}$P_*$}%
}}}}
\put(6676,-4711){\makebox(0,0)[lb]{\smash{{\SetFigFont{7}{8.4}{\familydefault}{\mddefault}{\updefault}{\color[rgb]{0,0,0}$u_2$}%
}}}}
\put(6676,-2086){\makebox(0,0)[lb]{\smash{{\SetFigFont{7}{8.4}{\familydefault}{\mddefault}{\updefault}{\color[rgb]{0,0,0}$u_1$}%
}}}}
\put(4876,-3736){\makebox(0,0)[lb]{\smash{{\SetFigFont{7}{8.4}{\familydefault}{\mddefault}{\updefault}{\color[rgb]{0,0,0}$P_\#$}%
}}}}
\end{picture}%

\end{center}
The cycle $z,y,u_1,u,u_2,z$ is completely labeled in $P_a$, and the
cycle $z,u_1,u_2,y,u,z$ is completely labeled in $P_b$. As before, the
interpretation of ${\rm sync}$ corresponds to all pairs in
$\{z,y,u_1,u,u_2\}$. The main difference with our previous version of
$q$ is that now there are no nodes linked by both edges $P_a$ and
$P_b$ in opposite directions. This implies that $q'$ can only have
parallel edges labeled $P_a$ (or $P_b$) and ${\rm sync}$ (in any
possible direction). This is crucial for our reduction to work.

 Since $q$ is now more complicated, we will have to modify $\Sigma$ in
order to ensure that $q$ maps into $\chase{q'}{\Sigma}$. In
particular, the third rule of $\Sigma$ must now ensure that the
structure of $q$ is completely replicated among the first element of
$q'$ (where the first element $x$ of $q$ will be mapped), the last element
of $q'$ (where the last element $v$ of $q$ will be mapped), and the
four elements that immediately precede the last element of $q'$ (where
the inner cycle of $q$ will be mapped).
The proof then mimicks the one we presented
before.

Let us assume now that $q'$ also admits loops. Since $q'$
homomorphically maps to $q$, these loops can only be labeled ${\rm
sync}$. Is this dangerous for our {\em backchase} analysis? Not
really. If one of these loops is used as a starting point for a chase
sequence, it can only mean that the synchronization of the words
in the PCP instance occurs earlier than expected. In particular, there is still a solution
for the instance.

The cases when the query $q$ has branching or disconnected components
is not more difficult, since in any case we can carry out the previous
analysis over one of the branches of $q$. This finishes the proof.

\subsection*{Proof of Lemma~\ref{lem:from-cq-to-acq}}

We first establish an auxiliary technical lemma:

\begin{lemma}\label{lem:small-acyclic-instance}
Let $q({\bar x})$ be a CQ, $I$ an acyclic instance, and ${\bar c}$ a tuple of constants. If there exists a homomorphism $h$ such that $h(q({\bar c})) \subseteq I$, then there exists an acyclic instance $J \subseteq I$, where $h(q({\bar c})) \subseteq J$ and $|J| \leqslant 2 \cdot |q|$.
\end{lemma}

\begin{proof}
Assume that $q$ is of the form $\exists {\bar y} \phi({\bar x},{\bar y})$.
Since $I$ is acyclic, there exists a join tree $T
= ((V,E),\lambda)$ of $I$. We assume, w.l.o.g., that for distinct nodes
$v,u \in V$, $\lambda(v) \neq \lambda(u)$. Let $T_q = ((V_q,E_q),\lambda)$ be the finite subforest of $T$ consisting of the nodes $\{v \in V \mid \lambda(v) \in h(\phi({\bar c},{\bar y}))\}$ and their ancestors. Let $F = ((V',E'),\lambda')$ be the forest obtained from $T_q$ as follows:
\begin{itemize}\itemsep-\parsep
\item $V' = \{v \in V_q \mid v \textrm{~is~either~a~root~node~or~a~leaf~node}\} \cup A$, where $A$ are the inner nodes of $T_q$ with at least two children;

\item For every pair of nodes $(v,u) \in V' \times V'$, $(v,u) \in
E'$ iff $u$ is a descendant of $v$ in $T_q$, and the unique shortest
path from $v$ to $u$ in $T_q$ contains only nodes of $((V \setminus
V') \cup \{v,u\})$; and

\item Finally, $\lambda' = \{x \mapsto y \mid x \mapsto y \in \lambda \textrm{~and~} x \in V'\}$, i.e., $\lambda'$ is the restriction of $\lambda$ on $V'$.
\end{itemize}

We define $J$ as the instance $\{\lambda'(v) \mid v \in V'\} \subseteq I$. It is clear that $h(\varphi({\bar c},{\bar y})) \subseteq J$.
Moreover, by construction, $|V'| \leqslant 2 \cdot |q|$, which in turn implies that $|J| \leqslant 2 \cdot |q|$. It remains to show that $J$ is acyclic, or,
equivalently, that $F$ is a join tree of $J$.
Since, by construction, $\{\lambda'(v) \mid v \in V'\} = J$, it remains to show that, for each term $t$ in $J$, the set $\{v \in V' \mid t \textrm{~occurs~in~} \lambda'(v)\}$ induces a connected subtree in $F$.
Consider two distinct nodes $v,u \in V'$ such that, for some $t$ in $J$, $t$ occurs in $\lambda'(v)$ and $\lambda'(u)$. By construction of $F$, there exists a path $v,w_1,\ldots,w_n,u$ in $F$ such that the nodes $w_1,\ldots,w_n$
occur in the unique path from $v$ to $u$ in $T$. Since $T$ is a join tree, $t$ occurs in $\lambda'(w_i)$, for each $i \in \{1,\ldots,n\}$. Hence, $F$ is a join tree of $J$, as needed.
\end{proof}

Having the above lemma in place, we can now establish Lemma~\ref{lem:from-cq-to-acq}.
Assume that $q$ is of the form $\exists {\bar y} \phi({\bar x},{\bar y})$. By hypothesis, there exists a homomorphism $h$ such that $h(\phi({\bar c},{\bar y})) \subseteq I$. By Lemma~\ref{lem:small-acyclic-instance}, there exists an acyclic instance $J \subseteq I$, where $h(\phi({\bar c},{\bar y})) \subseteq J$ and $|J| \leqslant 2 \cdot |q|$. For notational
convenience, let ${\bar c} = (c_1,\ldots,c_k)$. We define
$q'$ as the CQ $\exists {\bar w} \psi({\bar z},{\bar w})$, where
$|{\bar z}| = |{\bar x}|$, ${\bar z} = (V_{c_1},\ldots,V_{c_k}) \in
\variables^{k}$, and $\psi({\bar z},{\bar w})$ is the conjunction of
atoms $\bigwedge_{p({\bar u}) \in J} \, \rho(p({\bar u}))$, with
$\rho$ be a renaming substitution that replaces each term $t$ occurring in $J$ with the variable $V_t$. Intuitively, $q'$ is
obtained by converting $J$ into a CQ. Since, by hypothesis, $J$ is
acyclic, also $q'$ is acyclic. Clearly,
$\rho(h(\phi({\bar x},{\bar y}))) \subseteq \psi({\bar z},{\bar w})$ and
$\rho(h({\bar z})) = {\bar z}$, which implies that $q' \subseteq q$.
Moreover, since $|J| \leqslant 2 \cdot |q|$, $|q'| \leqslant 2 \cdot |q|$. Finally, observe that $\rho^{-1}(\psi({\bar z},{\bar w})) = J \subseteq I$ and
$\rho^{-1}({\bar z})= {\bar c}$, and therefore $q'({\bar c})$ holds in $I$, and the claim follows.

\subsection*{Proof of Proposition~\ref{pro:guarded-tgds-apc}}

Consider an acyclic CQ $q$, and a set $\dep \in \guarded$. We need to show that $\chase{q}{\dep}$, that is, the result of an arbitrary chase sequence
\[
q = I_0 \xrightarrow{\tau_0,\bar c_0} I_1 \xrightarrow{\tau_1,\bar c_1} I_2 \dots,
\]
for $q$ under $\dep$, admits a join tree. This can be done via the {\em guarded chase forest} for $q$ and $\dep$, which is defined as the labeled forest $F = (V,E,\lambda)$, where
\begin{enumerate}
\item $|V| = |\chase{q}{\dep}|$;
\item For each $R({\bar t}) \in \chase{q}{\dep}$, there exists a node $v$ such that $\lambda(v) = R({\bar t})$; and

\item The edge $(v,u)$ belongs to $E$ iff there exists $i \geqslant 0$ such that $\lambda(v) \in I_i$, the guard of $\tau_i$ is satisfied by $\lambda(v)$, and $\lambda(u) \in I_{i+1} \setminus I_i$.
\end{enumerate}

We proceed to show that each connected component of $F$, which is a tree with its root labeled by an atom $\alpha$ of $q$, is a join tree; we refer to this join tree by $T_{\alpha}$.
Fix an arbitrary atom $\alpha$ of $q$. We need to show that for each term $t$ occurring in $T_\alpha$, the set $\{v \in V \mid t \textrm{~occurs~in~} \lambda(v)\}$ induces a connected subtree in the guarded chase forest for $q$ and $\dep$.
Towards a contradiction, assume that the latter does not hold. This implies that there exists a path $v w_1 \ldots w_n u$ in the guarded chase forest for $q$ and $\dep$, where $n \geqslant 1$, and a term $t$ that occurs in $\lambda(v)$ and $\lambda(u)$, and $t$ does not occur in $\lambda(w_i)$, for each $i \in \{1,\ldots,n\}$. Assume that $\lambda(u)$ was generated during the $i$-th application of the chase step, i.e., $\lambda(u) \in I_{i+1} \setminus I_i$. Since $t$ does not occur in $\lambda(w_n)$, we conclude that $\sigma_i$ is not guarded. But this contradicts our hypothesis that $\dep \in \guarded$, and thus $T_\alpha$ is a join tree.

Since $q$ is acyclic it admits a join tree $T_q$. Let $T$ be the tree obtained by attaching $T_{\alpha}$ to the node of $T_q$ labeled by $\alpha$. Clearly, $T$ is a join tree for $\chase{q}{\dep}$, and the claim follows.

\subsection*{Guarded Tgds are not UCQ Rewritable}

Consider the guarded tgd
\[
\tau\ =\ P(x,y),S(x) \ra S(y)
\]
and the two Boolean CQs
\[
q\ =\ S(a) \wedge \phi_P \qquad q'\ =\ S(b),
\]
where $a,b$ are constants, and $\phi_P$ is a conjunction of atoms of the form $P(x,y)$, where $x,y$ are constants.
Assume there is a UCQ $Q$ such that $q \subseteq_{\{\tau\}} q'$ iff $Q(D_q) \neq \emptyset$, where $D_q$ consists of all the atoms in $q$.
This means that $Q$ is able to check for the existence of an unbounded sequence of atoms $P(a,c_1),P(c_1,c_2),\ldots,P(c_{n-1},b)$ in $D_q$. However, this is not possible via a finite (non-recursive) UCQ, which implies that $\guarded$ is not UCQ rewritable.

\subsection*{Proof of Proposition~\ref{pro:apc-kds-unary-binary-predicates}}

Consider an acyclic CQ $q$ over unary and binary predicates. It suffices to show that, after applying a key dependency $\epsilon$ of the form $R(x,y),R(x,z) \ra y=z$ on $q$, the obtained query $q_{\epsilon}$ is still acyclic.
Let $T_q$ be the join tree of $q$. Assume that $\epsilon$ is triggered due to the homomorphism $h$, i.e., $h$ maps the body of $\epsilon$ to $q$ and $h(y) \neq h(z)$. Without loss of generality, assume that the atom $h(R(x,y)$ is an ancestor of $h(R(x,z))$ in $T_q$. Let $\alpha$ be the first atom on the (directed) path from $h(R(x,y))$ to $h(R(x,z))$ in $T_q$ that contains both $h(x)$ and $h(z)$. Since we have only unary and binary predicates, we can safely conclude that all the atoms in $T_q$ that contain the term $h(z)$ belong to the subtree $T_\alpha$ of $T_q$ that is rooted on $\alpha$.
Therefore, if we delete the subtree $T_\alpha$ from $T_q$ and attach it on the atom $h(R(x,y)$, and then replace every occurrence of $h(z)$ with $h(y)$, we obtain a tree which is actually a join tree for $q_\epsilon$. This implies that $q_\epsilon$ is acyclic, and the claim follows.

\subsection*{Proof of Theorem \ref{theo:eval-guarded}}

We start by recalling the {\em existential 1-cover game} from
\cite{ChenD05}. This game is played by the {\em spoiler} and the {\em
  duplicator} on two pairs $(I,\bar t)$ and $(I',\bar t')$, where $I$
and $I'$ are instances and $\bar t$ and $\bar t'$ are two equally long
tuples of elements in $I$ and $I'$, respectively.
The game proceeds in
rounds. At each round either:
\begin{enumerate}
\item
The spoiler places a
pebble on an element $a$ of $I$, and the duplicator responds by
placing its corresponding pebble on an element $f(a)$ of $I'$, or
\item the spoiler removes one pebble from $I$,
and the duplicator responds by removing the
corresponding pebble from $I'$.
\end{enumerate}
The spoiler is constrained in the following way: (1) At any round $k$ of the
game, if $a_1,\dots,a_l$ ($l \leq k$)
are the elements covered by the pebbles of the spoiler in $I$, then
there must be an atom of $D$ that contains all such elements (this
explains why the game is called {\em 1-cover}, as there is always a single
atom that covers all elements which are pebbled), and (2)
if the spoiler places a pebble on the $i$-th component $t_i$ of $\bar
t$, then the duplicator must respond by placing the corresponding
pebble on the $i$-th component $t'_i$ of $\bar t'$.
The duplicator wins the game if he can always ensure that $f$ is a
{\em partial homomorphism} from $(a_1,\dots,a_l)$ in $I$ to
$I'$.\footnote{That is,
$f$ is a
homomorphism from $I(a_1,\dots,a_l)$ to $I'$, where $I(a_1,\dots,a_l)$
is the restriction of $I$
to those atoms $T(\bar a)$ such that $\bar a$ only mentions elements
in $\{a_1,\dots,a_l\}$.}
In such case we write $(I,\bar t)
\equiv_{\exists1c} (I',\bar t')$.

The following useful characterization of $(I,\bar t)
\equiv_{\exists1c} (I',\bar t')$ can be obtained from results in
\cite{ChenD05}:

\begin{lemma} \label{lemma:games} It is the case that $(I,\bar t) \equiv_{\exists1c}
  (I',\bar t')$ if and only if there is a mapping $\H$ that associates
  with each atom $T(\bar a)$ in $I$ a nonempty set $\H(T(\bar a))$ of
  atoms of the form $T(f(\bar a))$ in $I'$ and satisfies the following:
\begin{enumerate}
\item If the $i$-th component $a_i$ of $\bar a$ corresponds to the $j$-th
  component $t_j$ of $\bar t$, then for each tuple $T(f(\bar a)) \in
  \H(T(\bar a))$ it is the case that the $i$-th component of $f(\bar a)$
  corresponds to the $j$-th component $t'_j$ of $\bar t'$.
\item Consider an arbitrary atom $T(f(\bar a)) \in
  \H(T(\bar a))$. Then for each atom $S(\bar b)$ in $I$ there exists
  an atom $S(f'(\bar b))\in
  \H(S(\bar b))$ such that $f(c) = f'(c)$ for each element $c$ that
  appears both in $\bar a$ and $\bar b$.
\end{enumerate}
It follows, in particular, that for each tuple $T(f(\bar a)) \in
  \H(T(\bar a))$ the mapping $f$ is a partial homomorphism from $\bar
  a$ in $I$ to $I'$.
\end{lemma}

When such an $\H$ exists we call it a {\em winning strategy for the duplicator
  in the game on $(I,\bar t)$ and $(I',\bar t')$}. The existence of
a winning strategy for the duplicator can be decided in polynomial time over finite instances:

\begin{proposition} \label{prop:games-pol}
There exists a polinomial time algorithm that decides whether $(I,\bar
t) \equiv_{\exists1c} (I',\bar t')$, given finite instances $I$ and
$I'$ and tuples $\bar t$ and $\bar t'$ of elements in $I$ and $I'$,
respectively.
\end{proposition}

The following important fact can also be established from results in
\cite{ChenD05}:

\begin{proposition}  \label{prop:games-acyclic}
If $(I,\bar t) \equiv_{\exists1c} (I',\bar t')$, then for every
acyclic CQ $q$ it is the case that $\bar t' \in q(I')$ whenever
$\bar t \in q(I)$.
\end{proposition}

This implies, in particular, that for every instance $I$, tuple $\bar t$ of elements in
$I$, and CQ $q(\bar x)$ that is semantically acyclic (in the absence
of constraints), it is the case that $\bar t \in q(I)$ if and only if
$(q,\bar x) \equiv_{\exists1c} (I,\bar t)$\footnote{Here we slightly
  abuse notation and write $q$ for the
database that contains all the atoms of $q$.}. Applying Proposition
\ref{prop:games-pol} we obtain
that the evaluation of semantically acyclic CQs (in the absence of
constraints) is a tractable problem.

Now, assume that $q(\bar x)$ is semantically acyclic under a set
$\Sigma$ of tgds. Then the following holds:

\begin{proposition} \label{prop:games-chase}
 For every instance $I$
that satisfies $\Sigma$ and tuple $\bar t$ of elements in $I$, we have
that $\bar t \in q(I)$ if and only if $(\chase{q}{\Sigma},\bar x)
\equiv_{\exists1c} (I,\bar t)$.
\end{proposition}

\begin{proof}
Assume first that $\bar t \in
q(I)$. Then there is a homomorphism $h$ from $q$ to $I$ such that
$h(\bar x) = \bar t$. But since $I \models \Sigma$, it is easy to see
that $h$ extends to a homomorphism $h'$ from $\chase{q}{\Sigma}$ to
$I$ such that $h'(\bar x) = \bar t$. This implies, in particular, that
$(\chase{q}{\Sigma},\bar x) \equiv_{\exists1c} (I,\bar t)$ since the
duplicator can simply respond by following the homomorphism $h'$.
Assume, on the other hand, that $(\chase{q}{\Sigma},\bar x)
\equiv_{\exists1c} (I,\bar t)$.  Then Proposition
\ref{prop:games-acyclic} implies that for every acyclic CQ $q'$
we have that $\bar t \in q'(I)$ whenever $\bar x \in
q'(\chase{q}{\Sigma})$. We know that $q$ is equivalent to some acyclic
CQ $q^*$ under $\Sigma$, which implies that $\bar x \in
q^*(\chase{q}{\Sigma})$ from Lemma \ref{lemma:cq-equiv-tgds}. We
conclude then that $\bar t \in q^*(I)$, and, thus, that $\bar t \in
q(I)$ (since $q \equiv_\Sigma q^*$ and $I \models \Sigma$).
\end{proof}

Thus, in order to prove that $\saeval(\guarded)$ can be solved in
polynomial time, we only need to prove that the problem of checking
whether $(\chase{q}{\Sigma},\bar x) \equiv_{\exists1c} (D,\bar t)$ can
be solved in polynomial time, given a CQ $q$, a
database (finite instance) $D$ that satisfies a set $\Sigma$ of guarded tgds, and a
tuple $\bar t$ of elements in $D$. This is done by proving that if $\Sigma$ is guarded, then $(\chase{q}{\Sigma},\bar x) \equiv_{\exists1c}
(D,\bar t)$ if and only if $(q,\bar x) \equiv_{\exists1c} (D,\bar
t)$. Since we know from Proposition \ref{prop:games-pol} that
deciding the existence of a winning strategy for the duplicator
in the existential 1-cover game is in polynomial time, we conclude
that the problem of checking $(\chase{q}{\Sigma},\bar x) \equiv_{\exists1c}
(D,\bar t)$ can be solved efficiently. Thus, it only remains to prove
the following:

\begin{lemma} \label{lemma:gcg}
Let $\Sigma$ be a finite set of guarded tgds and $q(\bar x)$ a
CQ. Then for every database $D$ that satisfies $\Sigma$ and tuple $\bar t$ of elements in $D$,
it is the case that $(\chase{q}{\Sigma},\bar x) \equiv_{\exists1c}
(D,\bar t)$ if and only if $(q,\bar x) \equiv_{\exists1c} (D,\bar
t)$.
\end{lemma}

\begin{proof}
The implication from left to right is immediate since $q$ is
contained in $\chase{q}{\Sigma}$. Assume now that $(q,\bar x) \equiv_{\exists1c} (D,\bar
t)$. In virtue of Lemma \ref{lemma:games}, there is a winning strategy $\H$ for the duplicator in the game on
$(q,\bar x)$ and $(D,\bar t)$. We need to prove then that
there is a winning strategy $\H'$ for the duplicator in the game on
$(\chase{q}{\Sigma},\bar x)$ and $(D,\bar t)$. That is, that
there is a mapping $\H'$ that associates
  with each atom $T(\bar a)$ in $\chase{q}{\Sigma}$ a nonempty set $\H'(T(\bar a))$ of
  tuples of the form $T(f(\bar a))$ in $D$ and satisfies the following:
\begin{enumerate}
\item If the $i$-th component $a_i$ of $\bar a$ corresponds to the $j$-th
  component $x_j$ of $\bar x$, then for each tuple $T(f(\bar a)) \in
  \H'(T(\bar a))$ it is the case that the $i$-th component of $f(\bar a)$
  corresponds to the $j$-th component $t_j$ of $\bar t$.
\item Consider an arbitrary atom $T(f(\bar a)) \in
  \H'(T(\bar a))$. Then for each atom $S(\bar b)$ in $\chase{q}{\Sigma}$ there exists
  an atom $S(f'(\bar b))\in
  \H'(S(\bar b))$ such that $f(c) = f'(c)$ for each element $c$ that
  appears both in $\bar a$ and $\bar b$.
\end{enumerate}

Let us assume that $\chase{q}{\Sigma}$ is obtained by the following
sequence of chase steps:
\[
I_0 \xrightarrow{\tau_0,\bar c_0} I_1 \xrightarrow{\tau_1,\bar c_1} I_2 \dots
\]
We prove by induction that there are
mappings $(\H'_j)_{j \geq 0}$ such
that the following holds for each $j \geq 0$: (a) $\H'_j$ is a winning
strategy for the duplicator in the game on $(I_j,\bar x)$ and $(D,\bar
t)$, and (b) if $j > 0$ then $\H'_{j}(T(\bar a)) = \H'_{j-1}(T(\bar a))$ for each tuple
$T(\bar a) \in I_{j-1}$. This finishes the proof of Lemma \ref{lemma:gcg},
as it is clear then that $\H'$ can be defined as $\bigcup_{j \geq 0}
\H'_j$.

For the basis case $j = 0$ we have $I_0 = q$. Thus, we can define
$\H'_0$ to be $\H$. Let us consider then the inductive case $j + 1$,
for $j \geq 0$. By
inductive hypothesis, there is a winning strategy $\H'_j$ for the duplicator
in the game on $(I_j,\bar x)$ and $(D,\bar t)$.
Further, we have by definition that $I_{j+1}$ is the result of applying tgd
$\tau_j$ over $I_j$ with $\bar c_j$. Let us assume that $\tau_j$ is of
the form $\phi(\bar x) \rightarrow \exists \bar z \psi(\bar
y,\bar z)$, where $\bar y$ is a tuple of variables taken from $\bar x$. This
implies that $I_{j+1}$ extends $I_j$ with every tuple in $\psi(\bar
d_j,\bar z')$, where $\bar d_j$ is the restriction of
$\bar c_j$ to $\bar y$ and $\bar z'$ is a tuple that is obtained by replacing
each variable in $\bar z$ with a fresh null. We set
$\H'_{j+1}(T(\bar a)) := \H'_j(T(\bar a))$ for each atom $T(\bar a)$
in $I_j$. We explain next how to define $\H'_{j+1}$ over $I_{j+1}
\setminus I_j$.

Let $T(\bar a)$ be an atom in $I_{j+1} \setminus I_j$. This implies,
in particular, that $T(\bar a)$ belongs to $\psi(\bar d_j,\bar z')$.
Suppose that the guard of $\tau_j$ is $R(\bar x)$. Since $I_j \models
\phi(\bar c_j)$ we have that $R(\bar c_j)$ belongs to $I_j$, and
therefore that $\H'_j(R(\bar c_j))$ is well-defined and nonempty. Take
an arbitrary element $R(f(\bar c_j)) \in \H'_j(R(\bar c_j))$. Since
$\H'_j$ is a winning strategy for the duplicator in the game on $(I_j,\bar x)$ and
$(D,\bar t)$, we have from Lemma \ref{lemma:games} that $f$ is a partial homomorphism from $\bar
c_j$ in $I_j$ to $D$. This implies, in particular, that $D \models
\phi(f(\bar c_j))$ since $\tau_j$ is guarded.  But $D \models
\Sigma$, and therefore there is a tuple $g(\bar z')$ of elements in
$D$ such that $D \models \psi(f(\bar d_j),g(\bar z'))$. Let us define then
a mapping $h_f$ from $\bar a$ to $D$ such that for each $a$ in $\bar
a$ the value of $h_f(a)$ is defined as follows:
$$h_f(a) \ = \ \begin{cases}
f(a), \ \ \ & \text{if $a$ appears in $\bar d_j$,} \\
g(a), \ \ \ & \text{if $a$ appears in $\bar z'$.}
\end{cases}
$$
Notice, in particular, that $T(h_f(\bar a))$ belongs to $D$. We define
then $\H'_{j+1}(T(\bar a))$ as the set of all tuples in $D$ of the form
$T(h_f(\bar a))$, for $h$ a mapping such that $R(f(\bar c_j))$ belongs
to $\H'_j(R(\bar c_j))$. Clearly, $\H'_{j+1}(T(\bar a))$ is nonempty.

We prove next that $\H'_{j+1}$ satisfies the desired conditions:
(a) $\H'_{j+1}$ is a winning
strategy for the duplicator in the game on $(I_{j+1},\bar x)$ and $(D,\bar
t)$, and (b) $\H'_{j+1}(T(\bar a)) = \H'_j(T(\bar a))$ for each tuple
$T(\bar a) \in I_j$.  Condition (b) is satisfied by definition. We
concentrate on condition (a) now. Let us start with the first condition
in the definition of winning strategy. Take an arbitrary atom
$T(\bar a)$ in $I_{j+1}$ and assume that the $i$-th component
of $\bar a$
corresponds to the $j$-th component $x_j$ of $\bar x$. If $T(\bar a)$
also belongs to $I_j$, then we have by inductive hypothesis that for
each atom $T(f(\bar a)) \in \H'_{j+1}(T(\bar a))$ the $i$-th component
of $f(\bar a)$ corresponds to the $j$-th component $t_j$ of $\bar t$
(since $\H'_{j+1}(T(\bar a)) = \H'_j(T(\bar a))$ and $\H'_j$ is a
winning strategy for the duplicator in the game on $(I_{j},\bar x)$ and $(D,\bar
t)$). Let us assume, on the other hand, that $T(\bar a)$ belongs to $I_{j+1} \setminus
I_j$. In particular, $T(\bar a)$ belongs to $\psi(\bar d_j,\bar
z')$.
By definition of the chase, the fact that the $i$-th component $a$ of $\bar a$
corresponds
to the $j$-th component $x_j$ of $\bar x$ implies that $a$
belongs to $\bar d_j$ (as the elements in $\bar z'$ are fresh
nulls). Let us consider an arbitrary atom in $\H'_{j+1}(T(\bar
a))$. By definition, this atom is of the form $T(h_f(\bar a))$ for a
mapping $f$ such that $R(f(\bar c_j))$ belongs to $\H'_j(R(\bar
c_j))$. Since $a$ appears in $\bar d_j$ we have that $h_f(a) =
f(a)$. Further, by inductive hypothesis $f(a) = h_f(a)$ corresponds to $t_j$.

We now prove that the second condition in the definition of winning
strategy also holds for $\H'_{j+1}$. Let $T(\bar a)$ and $S(\bar b)$
be arbitrary atoms in $I_{j+1}$.
We prove that for each atom $T(h(\bar a)) \in \H'_{j+1}(T(\bar a))$
there is an atom $S(h'(\bar b)) \in \H'_{j+1}(S(\bar b))$ such that
$h$ and $h'$ coincide in all elements that are common to $\bar a$ and
$\bar b$.
We assume without loss of generality that
$\bar a$ and $\bar b$ have at least one element in common, otherwise
the property holds vacuously. We consider
two cases:

\begin{itemize}

\item $T(\bar a)$ and $S(\bar b)$ belong to $I_{j+1} \setminus I_j$.
This means that both
  $T(\bar a)$ and $S(\bar b)$ are atoms in $\psi(\bar d_j,\bar z')$
  that do not belong to $I_j$. Take an arbitrary atom in
  $\H'_{j+1}(T(\bar a))$. By definition, such atom is of the form
$T(h_f(\bar a))$ for a
mapping $f$ such that $R(f(\bar c_j))$ belongs to $\H'_j(R(\bar
c_j))$. But in the same way, then, $\H'_{j+1}(T(\bar a))$ must contain
an atom of the form $S(h_f(\bar b))$. This proves that the property
holds in this case.

\item Either $T(\bar a)$ or $S(\bar b)$ belongs to $I_j$. Suppose first
  that both $T(\bar a)$ and $S(\bar b)$ belong to $I_j$. Then
the property holds
  by inductive hypothesis (since $\H'_{j+1}(T(\bar a)) =
  \H'_{j}(T(\bar a))$, $\H'_{j+1}(S(\bar b)) = \H'_{j}(S(\bar b))$,
  and $\H'_j$ is a winning strategy for the duplicator in the game on
$(I_{j},\bar x)$ and $(D,\bar
t)$).

Let us assume without loss of generality then that $T(\bar a)
\in I_{j+1} \setminus I_j$ and $S(\bar b) \in I_j$. This means, in particular, that $T(\bar
a)$ belongs to $\psi(\bar d_j,\bar z')$ but not to $I_j$. Furthermore,
each element that is shared by $\bar a$ and $\bar b$ belongs to $\bar
d_j$ (as $\bar z'$ is a tuple of fresh nulls). Thus, each element
shared by $\bar a$ and $\bar b$ is also shared by $\bar c_j$ and
$\bar b$.
Consider first an
arbitrary atom in $\H'_{j+1}(T(\bar a))$. By definition, such atom is of the form
$T(h_f(\bar a))$ for a
mapping $f$ such that $R(f(\bar c_j))$ belongs to $\H'_j(R(\bar
c_j))$. But since $\H'_{j+1}(R(\bar c_j)) =
  \H'_{j}(R(\bar c_j))$, $\H'_{j+1}(S(\bar b)) = \H'_{j}(S(\bar b))$,
  and $\H'_j$ is a winning strategy for the duplicator in the game on
$(I_{j},\bar x)$ and $(D,\bar
t)$, we have that there is an atom $S(f'(\bar b)) \in \H'_{j+1}(S(\bar
b))$ such that $f$ and $f'$ coincide in the elements that are shared
by $\bar c_j$ and $\bar b$. This implies that $h_f$ and $f'$ coincide
in the elements that are shared by $\bar a$ and $\bar b$. Consider now
an atom $S(f'(\bar b)) \in \H'_{j+1}(S(\bar
b))$. Then there is an atom $R(f(\bar c_j)) \in \H'_{j+1}(R(\bar
c_j))$ such that $f$ and $f'$ coincide in all elements shared by $\bar
c_j$ and $\bar b$. Therefore, $h_f$ and $f'$ coincide in all elements
shared by $\bar a$ and $\bar b$. The property follows then since
$T(h_f(\bar a))$ belongs to $\H'_{j+1}(T(\bar a))$ by definition.
\end{itemize}

This concludes the proof of Lemma \ref{lemma:gcg}.
\end{proof}

\subsection*{Unions of Conjunctive Queries}

\begin{proposition}
Let $\Sigma$ be a finite set of tgds that belongs to a class that has acyclicity-preserving chase, and $Q$ a UCQ. If $Q$ is semantically acyclic under $\dep$, then, for each CQ $q \in Q$, (i) there exists an acyclic CQ $q'$, where $|q'| \leq 2 \cdot |q|$, such that $q \equiv_{\dep} q'$, or (ii) there exists $q'' \in Q$ such that $q \subseteq_{\dep} q''$.
\end{proposition}

\begin{proof}
Assume there exists $q \in Q$ such that (i) and (ii) do not hold. We need to show that $Q$ is not semantically acyclic. Towards a contradiction, assume there exists an acyclic UCQ $Q_A$ such that $Q \equiv_{\dep} Q_A$. Since $Q \subseteq_{\dep} Q_A$, there exists $q_A \in Q_A$ such that $q \subseteq_{\dep} q_A$. Moreover, since $Q_A \subseteq_{\dep} Q$, there exists $\hat{q} \in Q$ such that $q_A \subseteq_{\dep} \hat{q}$. Observe that $q = \hat{q}$; otherwise, $q \subseteq_{\dep} \hat{q}$ which contradicts the fact that (ii) does not hold. Therefore, $q \equiv_{\dep} q_A$, which in turn implies that $q$ is semantically acyclic under $\dep$. By Proposition~\ref{prop:apc-small-query-property}, we conclude that there exists an acyclic CQ $q'$, where $|q'| \leq 2 \cdot |q|$, such that $q \equiv_{\dep} q'$. But this contradicts the fact that (i) does not hold, and the claim follows.
\end{proof}

A similar result can be shown for UCQ rewritable classes of tgds. Notice that for the following result we exploit Proposition~\ref{prop:ucq-rewritability-small-query-property} instead of Proposition~\ref{prop:apc-small-query-property}.

\begin{proposition}
Let $\class{C}$ be a UCQ rewritbale class, $\dep \in \class{C}$ a finite set of tgds, and $q$ a UCQ. If $Q$ is semantically acyclic under $\dep$, then, for each CQ $q \in Q$, (i) there exists an acyclic CQ $q'$, where $|q'| \leq 2 \cdot f_{\class{C}}(q,\dep)$, such that $q \equiv_{\dep} q'$, or (ii) there exists $q'' \in Q$ such that $q \subseteq_{\dep} q''$.
\end{proposition}

By exploiting the above results, it is not difficult to show that the complexity of $\SA$ when we focus on UCQs under the various classes of sets of tgds considered in this work is the same as for CQs.

\end{document}